\documentclass{conm-p-l}

\usepackage{amsmath}
\usepackage{amsthm}
\usepackage{dsfont}
\usepackage{amsfonts}
\usepackage{tikz}
\usepackage{pgflibraryarrows}

\newcommand{\reals}{\mathbb{R}}

\newcommand{\complex}{\mathbb{C}}

\newcommand{\integers}{\mathbb{Z}}

\newcommand{\bracketb}[1]{\Big[#1\Big]}
\newcommand{\bracketc}[1]{\bigg[#1\bigg]}
\newcommand{\bracketd}[1]{\Bigg[#1\Bigg]}
\newcommand{\angles}[1]{\left\langle #1 \right\rangle}
\newcommand{\pb}[1]{\left\{#1\right\}}
\newcommand{\com}[1]{\left[#1\right]}

\newcommand{\paraa}[1]{\big(#1\big)}
\newcommand{\parab}[1]{\Big(#1\Big)}
\newcommand{\parac}[1]{\bigg(#1\bigg)}

\newcommand{\diag}{\operatorname{diag}}

\newcommand{\spacearound}[1]{\quad#1\quad}
\newcommand{\equivalent}{\spacearound{\Longleftrightarrow}}

\newtheorem{theorem}{Theorem}[section]
\newtheorem{proposition}[theorem]{Proposition}
\newtheorem{lemma}[theorem]{Lemma}
\newtheorem{definition}[theorem]{Definition}

\newtheorem{remark}[theorem]{Remark}

\newcommand{\xv}{\vec{x}}
\newcommand{\uv}{\vec{u}}
\newcommand{\vv}{\vec{v}}
\newcommand{\tr}{\operatorname{tr}}

\newcommand{\Et}{\tilde{E}}
\newcommand{\Dt}{\tilde{D}}

\newcommand{\deltat}{\tilde{\delta}}
\newcommand{\betat}{\tilde{\beta}}
\newcommand{\gammat}{\tilde{\gamma}}
\renewcommand{\deltat}{\tilde{\delta}}
\newcommand{\bt}{\betat}
\newcommand{\gt}{\gammat}
\newcommand{\dt}{\tilde{d}}
\newcommand{\kt}{\tilde{k}}
\newcommand{\mt}{\tilde{m}}
\newcommand{\nt}{\tilde{n}}
\newcommand{\et}{\tilde{e}}
\renewcommand{\mid}{\mathds{1}}
\newcommand{\C}{\mathcal{C}}
\newcommand{\Ch}{\hat{C}}
\newcommand{\ch}{\hat{c}}
\newcommand{\CLa}{\C_{L,a}}
\newcommand{\A}{\mathcal{A}}
\newcommand{\twomatrix}[4]{\begin{pmatrix}#1 & #2 \\ #3 & #4\end{pmatrix}}
\newcommand{\twovec}[2]{\begin{pmatrix}#1 \\ #2 \end{pmatrix}}
\newcommand{\Amatrix}{\twomatrix{\alpha}{\beta}{\gamma}{\delta}}
\newcommand{\half}{\frac{1}{2}}

\newcommand{\PEEt}{\mathcal{P}[E,\Et]}
\newcommand{\PDDt}{\mathcal{P}[D,\Dt]}
\newcommand{\Lh}{\hat{L}}
\newcommand{\Wd}{W^\dagger}
\newcommand{\vt}{\tilde{v}}
\newcommand{\Rtwoplus}{\reals^2_{>0}}
\newcommand{\Rtwopluszero}{\reals^2_{\geq 0}}
\newcommand{\lp}{\lambda_+}
\newcommand{\lm}{\lambda_-}

\newcommand{\Fpn}{F_+^n}
\newcommand{\Fmn}{F_-^n}

\tikzstyle{edge}=[color=black,line width=1px,arrows=-latex]
\tikzstyle{vertex}=[fill=black,shape=circle,inner sep=0px,outer sep=0px,minimum size=7px]
\def\vtextn(#1)#2{\draw (#1.north) node[anchor=south] {#2};}
\def\vtexts(#1)#2{\draw (#1.south) node[anchor=north] {#2};}
\def\vtexte(#1)#2{\draw (#1.east) node[anchor=west] {#2};}
\def\vtextw(#1)#2{\draw (#1.west) node[anchor=east] {#2};}
\def\vertex(#1)(#2){\path (#2) node[fill=black,outer sep=0cm,inner sep=0cm,minimum size=0.2cm,circle](#1) {};}
\def\edge(#1)(#2){\draw (#1) -- (#2);}
\def\diredge(#1)(#2){\draw [->] (#1) -- (#2);}

\numberwithin{equation}{section}

\title[Affine crossed product algebras and noncommutative
  surfaces]{Affine transformation crossed product like\\algebras and noncommutative
  surfaces}

\author{Joakim Arnlind}
\address{Albert Einstein Institute, Am M\"uhlenberg 1, D-14476 Golm, Germany.}
\email{joakim.arnlind@aei.mpg.de}

\author{Sergei Silvestrov}
\address{Center for Mathematical Sciences, Box 118, S-22100 Lund, Sweden.}
\email{sergei.silvestrov@math.lth.se}

\date{March, 2009}


\begin{document}

\maketitle

\begin{abstract}
  Several classes of $*$-algebras associated to the action of an
  affine transformation are considered, and an investigation of the
  interplay between the different classes of algebras is initiated.
  Connections are established that relate representations of $*$-algebras,
  geometry of algebraic surfaces, dynamics of affine transformations,
  graphs and algebras coming from a quantization procedure of
  Poisson structures.  In particular, algebras related to surfaces
  being inverse images of fourth order polynomials (in $\reals^3$) are
  studied in detail, and a close link between representation theory
  and geometric properties is established for compact as well as
  non-compact surfaces.
\end{abstract}

\section{Introduction}

\noindent The interplay between representation theory of $*$-algebras and
dynamical systems or more general actions of groups or semi-groups is
an expanding area of investigation deeply intertwined with origins of
quantum mechanics, foundations of invariants and number theory,
symmetry analysis, symplectic geometry, dynamical systems and ergodic
theory and several other parts of mathematics that are fundamental for
modern physics and engineering.  There has been three main
frameworks for investigation of such broad interplay. These frameworks
are intertwining greatly in terms of mathematical ideas, constructions
and goals, but developed to some extant independently in the last
sixty years due to historical and other reasons. One approach is based
on the systematic use of crossed product type operator algebras, that
is $C^*$-algebras and $W^*$-algebras, constructed as crossed products
of a "coefficient" algebra with a group (or more generally a
semi-group) acting on it.  In particular, for a topological space, an
algebra of continuous functions encodes properties of the space. The
dynamics given by iteration of transformations of the topological
space is encoded then by combining the commutative algebra of
continuous functions with the action into non-commutative
$C^*$-algebras or $W^*$-algebras with the product defined by a kind of
a generalized convolution twisted by the action.  Properties of the
action then correspond to properties of the corresponding crossed
product $C^*$-algebras and $W^*$-algebras and their
$*$-representations by operators on Hilbert spaces.  This approach can
historically be viewed as a vast extension of the theory of induced
representations of finite and compact groups on the one hand and as a
general abstract framework for foundations of quantum mechanics and
quantum field theory on the other hand \cite{EF1, EF2, EF3, Gl1, Gl2,
  Jorg-b-1, Jorg-art-2, MACbook1, MACbook2, MACbook3}.  In this
approach, representations of the corresponding $C^*$-algebras and
$W^*$-algebras are typically the $*$-representations by bounded
operators, a restriction inherited from the norm structures of
$C^*$-algebras and $W^*$-algebras. That restriction, while not
significant in some contexts such as for example those involving
dynamics or action on compact spaces, becomes an obstacle in the
context of quantum mechanics where unbounded operators and actions on
non-compact spaces play crucial role. Some classes of unbounded
operators are still manageable in this approach by some affiliation
procedures, that basically amounts to finding some specific functions
or other procedures making those unbounded operators into bounded ones
belonging to representations of some $C^*$-algebras and
$W^*$-algebras. Then, by working with these "bounded shadows" within
$C^*$-algebras and $W^*$-algebras, some properties of the affiliated
unbounded operators are traced back using the intrinsic properties of
the affiliation procedure.  This is however more an escape route rather
then a general approach, since there are often classes of
representations by unbounded operators, associated to corresponding
actions, that fall outside the applicability range of the specific
affiliation procedures. On the other hand, this approach based on
using $C^*$-algebras, $W^*$-algebras and more general Banach algebras,
without making specific choices of generators of the algebras, may be
viewed as a kind of non-commutative coordinate-independent approach
to simultaneous treatment of actions and spaces on the same level
within the same general framework. For references and further material
within this general context, see for example
\cite{ArchbordSpielberg,BR1,BR2,Dav,EF1,EF2,EF3,Gl1,Gl2,KTW,LiBingRen,Ped-book,Sak,STom4,
  svesildej0,SSD20,SSJ-1,svetom0,svetom20,Tak,Tom1,Tom2}.

The other framework is based on more direct analysis of operators
representing specific choices of generators for the algebras. This is
a more constructive "non-commutative coordinates" approach, as the
choice of generators can be viewed as a choice of non-commutative
coordinates.  This framework is often used in physics and engineering
models. The convenient choice of the generators (non-commutative
coordinates) as in any coordinate approach is a key to success of
further analysis. Typically, the generators satisfy some defining
commutation rules used then when multiplying various expressions and
functions of the generators. Choices of generators influence the form
and complexity of the corresponding commutation rules, with the best
choice of generators is often precisely that which makes dynamics or
actions appear explicitly when generators are intertwined in
computations using the commutation rules. The possibility to choose
such generators often means that the algebra itself might be presented
as some kind of generalized crossed product of another algebra by the
corresponding action or perhaps a quotient of such crossed
products. That shows the interplay and broad use of dynamics and
actions for construction, classification and properties of the
corresponding operators satisfying the commutation relations for the
generators.  Precisely as the coordinate approach is used in almost
any explicit applications and computational modeling in classical
mechanics and engineering problems, the non-commutative coordinates
approach based on generators and commutation relations is used
throughout quantum physics. Moreover, it is also used even in
classical mechanics and engineering for example in connection to
symmetry analysis of differential and difference equations.  This
generators and relations framework, while being slightly less general
than the coordinate-independent approach of working with
$C^*$-algebras and their representations, is more advantageous in
another important respect. Operators may satisfy commutation relations
in one or another sense without being bounded.  Such unbounded
families of operators might not be extendable to a representation of
the algebra.  Moreover, for unbounded operators typically (due to for
example Hellinger-Toeplitz Theorem from Functional analysis) the
domains of definitions are not the whole space, which might lead to
impossibilities to compose or take linear combinations of such
operators to form an image of a representation of the algebra.  This
is why operators satisfying commutation relations in one or another
sense are called representations of the commutation relations rather
then representations of the algebra generated by the generators and
relations.  The problem of extendibility of representations of
commutation relations to a representation of the algebra is then
considered for various classes of commutation relations, leading to
interesting and unexpected results and examples requiring development
of suitable function analytical and analytical methods. For some
material and further references and on interplay of crossed product
type algebras and dynamics and actions within the generators and
relations based framework we recommend for for further reading
\cite{OS1, OstSam-book, Sa, SaV, S-PhD, SW-1, VSa}.

The third framework is based on pure algebra and is closely related
algebraically to (coordinates independent) framework of crossed
product $C^*$-algebras and $W^*$-algebras, but typically not taking
into proper consideration norm or metric structures and thus often
excluding proper and complete study or even a possibility of
classifications or proper description of infinite-dimensional
representations. On the other hand, in the algebraic study of
representations of algebraic crossed product algebras, substantial
work has been done on general representations of the algebra which are
not necessarily $*$-representations. Also, while in approaches based
on $C^*$-algebras and $W^*$-algebras, by definition, algebras and as a
result also their representations are over complex or real numbers,
for algebraic crossed products all other kinds of fields are being
considered. For references and further material in this purely
algebraic context see for example \cite{K, NastVObook, NVO, Passman,
  johan1, johan2, johan3, johan4, johan5}.

In this article we will work within the second framework, as the
algebras we will consider are naturally defined by generators and
relations of a certain type closely linked to the action of general
affine transformations in two dimensions (see Definition
\ref{def:A_L}).  We establish close connection between these crossed
product-like algebras and algebras that arise from a quantization
procedure of Poisson brackets associated to a general class of
algebraic surfaces (see Definition \ref{def:CLa}, Proposition
\ref{prop:Casimir} and Section \ref{sec:SurfaceRelation}).  We will
mostly work in this article with finite-dimensional representations,
and also describe some classes of infinite-dimensional
representations. The algebras we consider are closely related to
crossed product algebras of the algebra of functions in two commuting
variables by the action of an additive group of integers or a
semi-group of non-negative integers via composition of a function with
powers of the affine transformation applied to the two-dimensional
vector of variables (see Remark \ref{rem:ALcrossedpr} and Proposition
\ref{ALfunctreordering}).  Therefore, there exists a strong interplay
between representations and especially $*$-representations of these
commutation relations and dynamics of the affine transformation (see
Sections \ref{sec:finitedimrepsCLa} and \ref{sec:repsAL}). Especially
the orbits play important role for all finite-dimensional
representations and also for some classes of infinite-dimensional
representations as these representations can be described explicitly
in terms of orbits or parts of orbits. Another way of expressing this
and the structure of representations is using graphs. In this paper,
representations of the algebras connected to affine transformation and
their structure is studied using both the orbits and the graphs of
iterations of the affine transformation.

One of our main goals in this paper is to establish and investigate
the interplay of representations of these parametric families of
commutation relations and algebras with the geometry of the
corresponding parametric families of algebraic surfaces.  In Sections
\ref{sec:SurfaceRelation} and \ref{sec:repsurf} we investigate what
happens with representations when a change in the parameters results
in a change of properties of the corresponding surface; e.g. from
compact to non-compact, from genus 0 to genus 1, changes in the number
connected components etc. These and various other aspects of the
interplay between geometry and representations are studied in detail.

\section{Two algebras related to an affine map}

\noindent Let us define an affine map $L:\reals^2\to\reals^2$ by
\begin{equation}\label{eq:Laffinemap}
  L(\xv) = A\xv + \uv\equiv
  \begin{pmatrix}
    \alpha & \beta \\ 
    \gamma & \delta 
  \end{pmatrix}
  \begin{pmatrix}
    x \\ y
  \end{pmatrix}
  +
  \begin{pmatrix}
    u \\ v
  \end{pmatrix},
\end{equation}
with $\alpha,\beta,\gamma,\delta,u,v\in\reals$. To every such affine
map we will associate two algebras.

\begin{definition} \label{def:A_L}
  Let $\complex\langle S,T,E,\Et\rangle$ be a free associative algebra
  on four letters over the complex numbers. Let $I$ be the two-sided
  ideal generated by the relations
  \begin{align}
    &\alpha ES+\beta\Et S+uS-SE=0\label{eq:AL1}\\
    &\gamma ES+\delta\Et S + vS -S\Et=0\label{eq:AL2}\\
    &\alpha TE+\beta T\Et + uT-ET=0\label{eq:AL3}\\
    &\gamma TE+\delta T\Et + vT -\Et T=0\label{eq:AL4}\\
    &E\Et-\Et E = 0,\label{eq:AL5}
  \end{align}
  where $\alpha,\beta,\gamma,\delta,u,v\in\reals$. We define
  $\A_L$ to be the quotient algebra $\complex\langle
  S,T,E,\Et\rangle\slash I$. We can also consider $\A_L$ to be a
  $\ast$-algebra by defining $S^\ast=T$, $T^\ast=S$, $E^\ast=E$ and
  $\Et^\ast=\Et$, since the set of relations \eqref{eq:AL1}--\eqref{eq:AL5}
  is invariant under this operation.
\end{definition}

\begin{remark} \label{rem:ALcrossedpr}
Note that the defining relations
\eqref{eq:AL1},\eqref{eq:AL2},\eqref{eq:AL3},\eqref{eq:AL4} and \eqref{eq:AL5}
of the algebras $\A_L$ can be written in the following form when rewritten using block matrix notation
\begin{align*}
\left(\begin{array}{cc} S & 0 \\ 0 & S \end{array}\right)
 \left(\begin{array}{c}E \\ \Et \end{array}\right) &=
 L\left(\left(\begin{array}{c}E \\ \Et \end{array}\right)\right)
 \left(\begin{array}{c} S \end{array}\right)  \\
 \left(\begin{array}{c}E \\ \Et \end{array}\right)
 \left(\begin{array}{c} T  \end{array}\right) &=
 \left(\begin{array}{cc} T & 0 \\ 0 & T \end{array}\right)
 L\left(\left(\begin{array}{c}E \\ \Et \end{array}\right)\right) \\
 E\Et-\Et E &= 0.
\end{align*}
This way of writing the relations indicates a close connection of
the algebra $\A_L$ to crossed product type algebras and hence interplay with dynamics of iterations of the
algebra $\A_L$ (see Proposition \ref{ALfunctreordering}).
\end{remark}

\begin{definition}\label{def:CLa}
  Let $\complex\angles{W,V}$ be a free associative algebra on two
  letters over the complex numbers, and let $L$ be the affine map on
  $\reals^2$ defined by $L\xv=A\xv+\uv$. For any $a\in\reals$, let
  $I_a$ be the two-sided ideal generated by the relations
  \begin{align}
    &W^2V = aW-(\det A)\,VW^2+(\tr A)\,WVW\label{eq:Cdef1}\\
    &WV^2 = aV-(\det A)\,V^2W+(\tr A)\,VWV\label{eq:Cdef2}.
  \end{align}
  We then define $\CLa$ to be the quotient algebra
  $\complex\angles{W,V}\slash I_a$. We can also consider $\CLa$ to be a
  $\ast$-algebra by defining $W^\ast=V$ and $V^\ast=W$, since the set
  of relations \eqref{eq:Cdef1}--\eqref{eq:Cdef2} is invariant under
  this operation.
\end{definition}

\noindent In order to relate these algebras, we want to construct a
homeomorphism $\psi$ from $\A_L$ to $\C_L$, by setting
\begin{align*}
  &\psi(S) = W;\quad \psi(T)=V\\
  &\psi(E) = k\mid + mWV + nVW\\
  &\psi(\Et ) = \kt\mid + \mt WV + \nt VW.
\end{align*}
To obtain a homeomorphism, we must require that elements that are
equivalent to $0$ in $\A_L$ are mapped to elements equivalent to $0$
in $\CLa$. This requirement gives rise to the following system of equations
\begin{align*}
  &\Amatrix
  \twomatrix{m}{n}{\mt}{\nt}-
  \twomatrix{m}{n}{\mt}{\nt}
  \twomatrix{\tr A}{-\det A}{1}{0}=0\\
  &\bracketc{\Amatrix-\mid_2}
  \twovec{k}{\kt} = 
  \twomatrix{m}{n}{\mt}{\nt}\twovec{a}{0}-\twovec{u}{v}.
\end{align*}
General solutions to this system of equations are given in Appendix A,
but whenever $a\neq 0$, a particularly simple solution is given by
\begin{align*}
  &\psi(E) = \frac{1}{a}\bracketb{uWV + \paraa{\beta v-\delta u}VW}\\
  &\psi(\Et ) = \frac{1}{a}\bracketb{vWV + \paraa{\gamma u-\alpha v}VW}.  
\end{align*}
 The fact that
$\psi([E,\Et])=0$ is guaranteed by the following proposition.
\begin{proposition}\label{prop:DDtcommute}
  In $\CLa$ it holds that $[WV,VW]=0$.
\end{proposition}
\begin{proof}
  Multiplying \eqref{eq:Cdef1} by $V$ from the left and
  \eqref{eq:Cdef2} by $W$ from the right gives $VW^2V=WV^2W$,
  i.e. $[WV,VW]=0$.
\end{proof}
\noindent The map $\psi$ will in general not be an isomorphism since, e.g., the
element $E-k\mid-mST-nTS$ (which is non-zero in $\A_L$) is mapped to
$0$ in $\CLa$.

\section{The center of $\A_L$ and $\CLa$}

\noindent Let $\PEEt$ denote the subalgebra of $\A_L$ generated by $E$
nd $\Et$, and let $\PDDt$ denote the subalgebra of $\CLa$ generated by
$D=WV$ and $\Dt=VW$. In this section we will gather a couple of
results that concern central elements in $\A_L$ and $\CLa$.

\begin{proposition}\label{ALfunctreordering}
  For any $p\in\PEEt$ it holds that
  \begin{align*}
    &S^np(E,\Et) = p\paraa{L^n(E,\Et)}S^n\\
    &p(E,\Et)T^n = T^np\paraa{L^n(E,\Et)}
  \end{align*}
  where $L(x,y)=(\alpha x+\beta y+u,\gamma x+\delta y+v)$.
\end{proposition}

\begin{proposition}\label{prop:CLaCommutation}
  For any $p\in\PDDt$ it holds that
  \begin{align*}
    &W^np(D,\Dt) = p\paraa{\Lh^n(D,\Dt)}W^n\\
    &p(D,\Dt)V^n = V^np\paraa{\Lh^n(D,\Dt)}
  \end{align*}
  where $\Lh(x,y)=\paraa{(\tr A)x-(\det A)y+a,x}$.  
\end{proposition}

\noindent From these propositions it is clear that any polynomial $p$,
satisfying $p(x,y)=p\paraa{\Lh(x,y)}$ and any polynomial $q$,
satisfying $q(x,y)=q\paraa{L(x,y)}$ generate central elements of $\CLa$
and $\A_L$ respectively. In particular, we have the following result
\begin{proposition}\label{prop:Casimir}
  Let $\Ch_{r,s,t}$ denote the following element in $\CLa$:
  \begin{align*}
    \Ch_{r,s,t}=r\paraa{D+\Dt}+s\paraa{D+\Dt}^2+t\paraa{D-\Dt}^2.
  \end{align*}
  Then $\Ch_{r,s,t}$ commutes with $W$ and $V$ if and only if we are
  in one of the following two situations:
  \begin{enumerate}
  \item $\det A=1$, which implies that
    \begin{align}\label{eq:Casimir}
      \Ch = -4a\paraa{D+\Dt}+(2-\tr A)\paraa{D+\Dt}^2+(2+\tr A)\paraa{D-\Dt}^2
    \end{align}
    commutes with $W$ and $V$;
  \item $\det A=-1$, $\tr A=0$ and $a=0$, in which case $\Ch_{r,s,t}$
    commutes with $W$ and $V$ for all $r,s,t\in\reals$.
  \end{enumerate}
\end{proposition}

\section{Relation to noncommutative surfaces}\label{sec:SurfaceRelation}

\noindent In \cite{abhhs:noncommutative,a:repCalg} noncommutative
$C$-algebras of Riemann surfaces were constructed and a particular
case of spheres and tori was studied in detail. It turns out that the
classical transition from spherical to toroidal geometry corresponds
to a change in the representation theory of the noncommutative
algebras. This correspondence will later be described in detail. Let
us briefly recall how to obtain algebras from a given surface.

Let $C(x,y,z)$ be a polynomial and let $\Sigma=C^{-1}(0)$. One can
define a Poisson bracket on $\reals^3$ by setting
\begin{align}
  \pb{f,g} = \nabla C\cdot\paraa{\nabla f\times\nabla g},
\end{align}
for smooth functions $f,g$. This Poisson bracket induces a Poisson
bracket on $\Sigma$ by restriction. The idea is to start from the
coordinate relations 
\begin{align}
  &\pb{x,y} = \partial_zC\\
  &\pb{y,z} = \partial_xC\\
  &\pb{z,x} = \partial_yC
\end{align}
and then construct a noncommutative algebra on $X,Y,Z$ by imposing the
relations
\begin{align}
  &\com{X,Y}=i\hbar\Psi\paraa{\partial_z C}\\
  &\com{Y,Z}=i\hbar\Psi\paraa{\partial_x C}\\
  &\com{Z,X}=i\hbar\Psi\paraa{\partial_y C}
\end{align}
where $\Psi$ is an ordering map from polynomials in three variables to
noncommutative polynomials in $X$, $Y$ and $Z$. In case this algebra
is non-trivial, its representations will provide an approximating
sequence (in the sense of \cite{bhss:glinfinity}) for the Poisson
algebra of polynomial functions on $\Sigma$ as $\hbar\to 0$ (see
\cite{a:phdthesis} for details).

Let us now consider the following polynomial
\begin{align}
  C(x,y,z) = \frac{\alpha_0}{2}\paraa{x^2+y^2}+\frac{\alpha_1}{4}\paraa{x^2+y^2}^2
  +\half z^2-\half c_0,\label{eq:Cstd}
\end{align}
which, by using the above Poisson bracket, gives rise to
\begin{align}
  &\pb{x,y} = z\\
  &\pb{y,z} = \alpha_0x+\alpha_1x(x^2+y^2)\\
  &\pb{z,x} = \alpha_0y+\alpha_1y(x^2+y^2).
\end{align}
We will choose an ordering of the right hand sides in terms of the
complexified variables $W=X+iY$ and $V=X-iY$ (cp. \cite{a:repCalg})
\footnote{Note that in \cite{a:repCalg}, the parameter $\deltat_1$ is not
  present (although it is implicitly present in
  \cite{abhhs:noncommutative}, taking the value $1/2$). This is an additional freedom in the
  choice of ordering that can not be extended to higher order algebras
  without breaking the commutativity of $WV$ and $VW$.}
\begin{align*}
  &[X,Y] = i\hbar Z\\ 
  &[Y,Z] = i\hbar\alpha_0 X+\frac{i\hbar}{2}\bracketc{\bt_1\paraa{V^2W+VW^2}+
    \gt_1\paraa{VWV+WVW}+\deltat_1\paraa{WV^2+W^2V}}\\ 
  &[Z,X] = i\hbar\alpha_0 Y+\frac{i\hbar}{2i}\bracketc{\bt_1\paraa{VW^2-V^2W}+
    \gt_1\paraa{WVW-VWV}+\deltat_1\paraa{W^2V-WV^2}}
\end{align*}
for any choice of $\bt_1,\gt_1,\deltat_1$ such that
$\bt_1+\gt_1+\deltat_1=\alpha_1$. By eliminating $Z=[X,Y]/i\hbar$, one can
write the second two relations entirely in terms of $W$ and $V$
\begin{align*}
  &\paraa{1+2\hbar^2\deltat_1}W^2V = -2\alpha_0\hbar^2W - \paraa{2\hbar^2\bt_1+1}VW^2+\paraa{2-2\hbar^2\gt_1}WVW\\
  &\paraa{1+2\hbar^2\deltat_1}WV^2 = -2\alpha_0\hbar^2V - \paraa{2\hbar^2\bt_1+1}V^2W+\paraa{2-2\hbar^2\gt_1}VWV.
\end{align*}
This algebra is isomorphic to $\CLa$ if
\begin{align*}
  a=-\frac{2\alpha_0\hbar^2}{1+2\hbar^2\deltat_1}  
\end{align*}
and $L$ is an affine map such that
\begin{align*}
  \det A = \frac{1+2\hbar^2\bt_1}{1+2\hbar^2\deltat_1}
  \qquad\tr A = \frac{2-2\hbar^2\gt_1}{1+2\hbar^2\deltat_1}.
\end{align*}
Hence, the relation to the original parameters of the polynomial is
\begin{align*}
  \alpha_0 = -a\frac{1+2\hbar^2\deltat_1}{2\hbar^2}\qquad
  \alpha_1 = \Delta\frac{1+2\hbar^2\deltat_1}{2\hbar^2}
\end{align*}
where $\Delta = 1+\det A - \tr A$.

Let us study the Casimir $\Ch$, defined in \eqref{eq:Casimir} when
$\det A=1$, by writing it in terms of $X$, $Y$ and $Z$. Since
$D+\Dt=2\paraa{X^2+Y^2}$ and $D-\Dt=2\hbar Z$, we obtain
\begin{align*}
  \Ch = -8a\paraa{X^2+Y^2}
  +4\paraa{2-\tr A}\paraa{X^2+Y^2}^2+4\hbar^2\paraa{2+\tr A}Z^2. 
\end{align*}
When the algebra $\CLa$ arises from a surface, we can express $\tr A$
in terms of $\alpha_0,\alpha_1,\bt_1,\gt_1,\deltat_1$ to obtain
\begin{align*}
  \frac{1+2\hbar^2\deltat_1}{16\hbar^2}\Ch = \alpha_0\paraa{X^2+Y^2}
  +\frac{\alpha_1}{2}\paraa{X^2+Y^2}^2
  +\parab{1+2\hbar^2\deltat_1-\half\hbar^2\alpha_1}Z^2. 
\end{align*}
In this way we see that the Casimir $\Ch$ is a noncommutative analogue
of the embedding polynomial $C(x,y,z)$. In any irreducible
representation $\phi$, the element $\phi(\Ch)$ will be proportional to
the identity. Let us define two constants $\ch_0$ and $\ch_1$ through the following
relations:
\begin{align*}
  \phi(\Ch) = 4\ch_1\mid
   \qquad\text{and}\qquad 
   \ch_0 = \frac{1+2\hbar^2\deltat_1}{4\hbar^2}\ch_1.
\end{align*}
In the procedure of constructing noncommutative algebras from a given
polynomial, information about the constant $c_0$ is lost since the
construction only depends on partial derivatives of $C(x,y,z)$. As we
will see, different values of $c_0$ correspond to, for instance,
different topologies of the surface, and this rises a problem if
we want to study geometry in the algebraic setting. However, since
(when $\det A=1$) the central element $\Ch$ is a noncommutative analogue
of the polynomial $C(x,y,z)$, we will identify $c_0$ and $\ch_0$ in an
irreducible representation; this gives us a way to determine the
``topology'' of a representation.

In the following we will compare the geometry of the surface, for all
values of $\alpha_0,\alpha_1,c_0$, with the representation theory for
the corresponding irreducible representations of $\CLa$ when
$\ch_0=c_0$.

\section{$\ast$-representations of $\CLa$}\label{sec:finitedimrepsCLa}

\noindent From the viewpoint of noncommutative surfaces, one is
interested in representations in which $X,Y$ and $Z$ are self-adjoint
operators. This requirement is transferred to $\CLa$ by considering
$\ast$-representations.  By a $\ast$-representation we mean a
representation $\phi$ such that $\phi(A^*)=\phi(A)^\dagger$. Clearly,
writing $W=X+iY$ and $V=X-iY$, for hermitian matrices $X,Y$, implies
that $W^\dagger=V$. 

The ($\ast$-)representation theory of $\CLa$ was worked out in
\cite{a:repCalg}, but let us recall some details in the
construction. Let us for simplicity denote $\phi(W)$ by $W$ and
$\phi(V)$ by $V$ in a finite dimensional $\ast$-representation of
$\CLa$. By Proposition \ref{prop:DDtcommute} the matrices $D=W\Wd$ and
$\Dt=\Wd W$ will be two commuting hermitian matrices, and therefore
they can always be simultaneously diagonalized by a unitary
matrix. Let us assume such a basis to be chosen and write
\begin{align*}
  &D = \diag(d_1,d_2,\ldots,d_n)\\
  &\Dt = \diag(\dt_1,\dt_2,\ldots,\dt_n).
\end{align*}
In components, the defining relations of $\CLa$ (together with the
associativity condition $DW=W\Dt$) can then be written as
\begin{align*}
  &W_{ij}\parab{\paraa{\tr A}d_i-\paraa{\det A}\dt_i+a-d_j}=0\\
  &W_{ij}\paraa{d_i-\dt_j}=0,
\end{align*}
Thus, either $W_{ij}=0$ or
\begin{align*}
  d_j &= \paraa{\tr A}d_i-\paraa{\det A}\dt_i+a\\
  \dt_j &= d_i.
\end{align*}
By introducing the notation $\xv_i=(d_i,\dt_i)$ and the affine map $\Lh$
\begin{align*}
  \Lh\twovec{x}{y}=\twomatrix{\tr A}{-\det A}{1}{0}\twovec{x}{y}+\twovec{a}{0},
\end{align*}
we can write this relation as $\xv_j=\Lh(\xv_i)$ whenever $W_{ij}\neq
0$. Let us now show how the representation theory can be described as
a dynamical system generated by $\Lh$ acting on a directed graph.

Let $G_W=(V,E)$ be the directed graph of $W$, i.e. the graph on $n$
vertices with vertex set $V=\{1,2,\ldots,n\}$ and edge set $E\subseteq
V\times V$, such that
\begin{align*}
  (i,j)\in E\equivalent W_{ij}\neq 0.
\end{align*}
By assigning the vector $\xv_i$ to the vertex $i$, it follows that for
a graph corresponding to the matrix $W$ in a representation of $\CLa$,
it holds that $\xv_j=\Lh(\xv_i)$ whenever there is an edge from $i$ to
$j$. The dynamical system on the graph can therefore be depicted as in Figure \ref{fig:repgraph}.

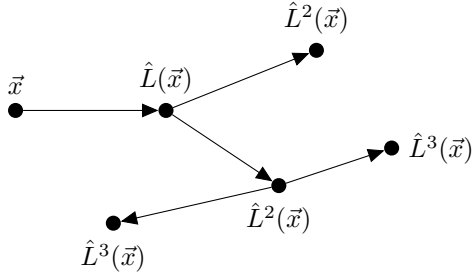
\begin{figure}[h]\label{fig:repgraph}
\begin{center}
  \begin{tikzpicture}[>=triangle 45]
    \vertex(node1)(0,0)
    \vertex(node2)(2,0)
    \vertex(node3)(4,0.8)
    \vertex(node4)(3.5,-1)
    \vertex(node5)(1.3,-1.5)
    \vertex(node6)(5,-0.5)
    \diredge(node1)(node2)
    \diredge(node2)(node3)
    \diredge(node2)(node4)
    \diredge(node4)(node5)
    \diredge(node4)(node6)
    \vtextn(node1){$\xv$}
    \vtextn(node2){$\Lh(\xv)$}
    \vtextn(node3){$\Lh^2(\xv)$}
    \vtexts(node4){$\Lh^2(\xv)$}
    \vtexts(node5){$\Lh^3(\xv)$}
    \vtexte(node6){$\Lh^3(\xv)$}
  \end{tikzpicture}
\end{center}
\caption{The affine map $\Lh$ acting on the directed graph of a representation.}
\end{figure}

\noindent One immediate observation is that if the graph has a
``loop'' (i.e. a directed cycle) on $k$ vertices, then the affine map
must have a periodic orbit of order $k$. If the affine map does not
have any periodic points, then loops are excluded from all representation
graphs. It is a trivial fact that any finite directed graph must have
at least one loop or at least one ``string'', i.e. a directed path
from a transmitter to a receiver. Hence, if the graph can not have a
loop, it must contain a string. Due to the fact that $D=W\Wd$ and
$\Dt=\Wd W$, one gets the following condition for vertices being
transmitters or receivers.
\begin{lemma}[\cite{abhhs:noncommutative}]\label{lemma:transrec}
  The vertex $i$ is a transmitter if and only if $\dt_i=0$. The
  vertex $i$ is a receiver if and only if $d_i=0$.
\end{lemma}
\noindent Thus, for a string on $k$ vertices to exist, there must
exist a vector $\xv=(d,0)$ such that $\Lh^{k-1}(\xv)=(0,\dt)$. We call
this a $k$-string of the affine map $\Lh$. We also note that since the
matrices $D$ and $\Dt$ are non-negative, all vectors $\{\xv_i\}$ must
lie in $\Rtwopluszero=\{(x,y)\in\reals^2: x,y\geq 0\}$. The natural question
is now: Which graphs correspond to irreducible representations of
$\CLa$? The answer lies in the following theorem.

\begin{theorem}[\cite{a:repCalg}]\label{thm:CLaRep}
  Let $\phi$ be a locally injective $\ast$-representation of
  $\CLa$. Then $\phi$ is unitarily equivalent to a representation in
  which $\phi(WV)$ and $\phi(VW)$ are diagonal and the directed graph
  of $\phi(W)$ is a direct sum of strings and loops. A representation
  corresponding to a single string or a single loop is irreducible.
\end{theorem}
 
\begin{remark}\label{rem:locallyInjective}
  A representation is locally injective if $\Lh$ is injective on the
  set $\{\xv_1,\xv_2,\ldots,\xv_n\}$. A representation whose graph is
  connected and contains a loop will always be locally injective
  \cite{a:repCalg}. Clearly, if $\Lh$ is invertible, then any
  representation is locally injective.
\end{remark}

\noindent Furthermore, one can show that every $k$-string in
$\Rtwopluszero$ and every periodic orbit in $\Rtwoplus=\{(x,y)\in\reals^2:
x,y>0\}$ induce an irreducible representation of $\CLa$; distinct
orbits/$k$-strings induce inequivalent representations. In this way,
the representation theory of $\CLa$ is completely determined by the
dynamical properties of the affine map $\Lh$.

For instance, assume that $\Lh^n(\xv_1)=\xv_1$ and
$\xv_k=\Lh^{k-1}(\xv_1)=(d_k,\dt_k)\in\Rtwoplus$ for
$k=1,\ldots,n-1$. Then an $n$-dimensional $\ast$-representation of
$\CLa$ is constructed by setting
\begin{align}
  \phi(W) = 
  \begin{pmatrix}
    0               &  \sqrt{d_1}  & 0                & \cdots & 0 \\
    0               &  0                & \sqrt{d_2}  & \cdots & 0 \\
    \vdots          & \vdots            & \ddots           & \ddots & \vdots\\
    0               & 0                 & \cdots           & 0 &  \sqrt{d_{n-1}} \\
    e^{i\beta}\sqrt{d_n}               & 0                 & \cdots           & 0 & 0
  \end{pmatrix}
\end{align}
for any $\beta\in\reals$.

\subsection{Infinite dimensional representations}

There are two classes of infinite dimensional representations of
$\CLa$ that can be easily constructed. They come in the form of
infinite dimensional matrices with a finite number of non-zero
elements in each row and column. This assures that the usual matrix
multiplication is still well-defined. 

The first type is \emph{one-sided infinite dimensional
  representations}. In this case the basis of the vector space is
labeled by the natural numbers. The second type is \emph{two-sided}
infinite dimensional representations; and the basis vectors are
labeled by the integers. 

A one-sided representation of $\CLa$ can be constructed by choosing
$\xv_0=(d_0,0)$ (with $d_0>0$) such that
$\xv_k=(d_k,\dt_k)=\Lh^{k}(\xv_0)\in\Rtwoplus$ for $k=1,2,\ldots$. A
one-sided representation is then obtained by letting $\phi(W)$ be an
infinite dimensional matrix with non-zero elements
$W_{k,k+1}=\sqrt{d_k}$ for $k=0,1,\ldots$. 

If we assume $\Lh$ to be invertible, two-sided representations can be
constructed by choosing $\xv_0\in\Rtwoplus$ such that
$\xv_k=(d_k,\dt_k)=\Lh^k(\xv_0)\in\Rtwoplus$ for $k\in\integers$. We then set the
non-zero elements of $\phi(W)$ to be $W_{k,k+1}=\sqrt{d_k}$ for $k\in\integers$.

\subsection{Representations when $\det A=1$}\label{sec:CLaDetOne}

Let us now turn to the question concerning when different kinds of
representations can exist, if we fix a specific value of the central
element $\Ch$. Thus, we assume that $\det A=1$ and that the
irreducible representation is such that
\begin{align}
  -4a(D+\Dt)+(2-\tr A)(D+\Dt)^2+(2+\tr A)(D-\Dt)^2=4\ch_1\mid.
\end{align}
Since $D$ and $\Dt$ are diagonal, this constrains the vectors
$\xv_i=(d_i,\dt_i)$ to lie in the set defined by all $(r,s)\in\reals^2$ such that 
\begin{align}\label{CLaCurve}
  p(r,s)\equiv 
  -4a(r+s)+(2-\tr A)(r+s)^2+(2+\tr A)(r-s)^2-4\ch_1=0
\end{align}
We define $\Gamma=p^{-1}(0)$ and call this set the \emph{constraint
  curve} of an irreducible representation. When $\tr
A\neq 2$, we can write the constraint curve in the following convenient
form
\begin{align*}
  p(r,s)=\Delta\parac{(r+s-2\mu)^2+\frac{2+\tr A}{\Delta}(r-s)^2-4\ch}
\end{align*}
where $\Delta=2-\tr A$, $\mu=a/\Delta$ and $\ch=\mu^2+\ch_1/\Delta$.
By Proposition \ref{prop:CLaCommutation}, $\Gamma$ is invariant under
the action of the affine map $\Lh$. Moreover, if $\Gamma$ has several
disjoint components, $\Lh$ leaves each of them invariant. 

In the case when $\Gamma$ consists of one or two disjoint curves, one
can check that $\Lh$ will preserve the direction along the curves;
i.e. we can parametrize each curve by $\gamma$ and denote points on
the curve by $\xv(\gamma)$, such that if we define $\gamma_1'$ and
$\gamma_2'$ through $\xv(\gamma_1')=L(\xv(\gamma_1))$ and
$\xv(\gamma_2')=L(\xv(\gamma_2))$ then $\gamma_1'\geq\gamma_1$, and
$\gamma_1\geq\gamma_2$ implies that $\gamma_1'\geq\gamma_2'$.

Let us now prove a few results leading to Proposition
\ref{prop:CLaOneDimRep}, that tells us when there are no non-trivial
(i.e. of dimension greater than one) finite dimensional
representations.

\begin{proposition}\label{prop:Lhit}
  Let $\Lh$ denote the affine map defined by $\Lh(x,y)=(x\tr A-y\det
  A+a,x)$ and assume that $(\tr A)^2\neq 4\det A$ and $\tr A\neq 1+\det
  A$. Then it holds that
  \begin{align*}
    \Lh^n\twovec{x}{y}=
    \frac{1}{\lp-\lm}
    &\twomatrix{\lp^{n+1}-\lm^{n+1}}{-q\paraa{\lp^n-\lm^n}}{\lp^n-\lm^n}{-q\paraa{\lp^{n-1}-\lm^{n-1}}}
    \twovec{x}{y}\\
    +&\frac{a}{\lp-\lm}
    \twovec{\lp\frac{1-\lp^n}{1-\lp}-\lm\frac{1-\lm^n}{1-\lm}}{\frac{1-\lp^n}{1-\lp}-\frac{1-\lm^n}{1-\lm}}.
  \end{align*}
  where $\lambda_\pm=\parab{\tr A\pm\sqrt{(\tr A)^2-4\det A}}/2$.
\end{proposition}

\begin{lemma}\label{lemma:LhNoPeriods}
  Assume that $\det A=1$ and $\Delta\leq 0$. Then $\Lh$ has no periodic
  points other than fix-points.
\end{lemma}

\begin{proof}
  When $\tr A=2$ and $\det A=1$, a direct computation shows that there
  are only periodic points when $a=0$, and these points are
  fix-points. 

  When $\tr A>2$ and $\det A=1$, the eigenvalues of the matrix
  \begin{align*}
    M=\twomatrix{\tr A}{-1}{1}{0}
  \end{align*}
  are real and distinct; furthermore, they are both different from
  $\pm 1$. Since no eigenvalue equals $1$, the affine map $\Lh$ is
  equivalent to the \emph{linear} map $M$ around some point. Thus,
  finding periodic points of $\Lh$ is equivalent to finding periodic
  points of $M$. Moreover, since the eigenvalues of $M$ are distinct,
  the matrix is diagonalizable. In total, this means that periodic
  points (of period greater than one) of $\Lh$ exist if and only if
  one of the eigenvalues of $M$ is an $n$'th root of unity. But this
  is impossible since both eigenvalues are real and different from
  $\pm 1$. Hence, $\Lh$ has no periodic points except for the
  possible fix-points.
\end{proof}

\begin{lemma}\label{lemma:nostring1}
  Assume that $\det A=1$, $\Delta<0$ and $a\geq 0$. For any integer $n\geq 1$,
  there are no $x,y>0$ such that $\Lh^n(x,0)=(0,y)$.
\end{lemma}

\begin{proof}
  When $\det A=1$ and $\tr A>2$, the relations $(\tr A)^2\neq 4\det A$
  and $\tr A\neq 1+\det A$ are fulfilled. Therefore, by Proposition
  \ref{prop:Lhit}, $\Lh^n(x,0)=(0,y)$ is equivalent to
  \begin{align}
    &\paraa{\lp^{n+1}-\lm^{n+1}}x+a\paraa{\lp\Fpn-\lm\Fmn} = 0\label{eq:Lhit1}\\
    &\paraa{\lp^n-\lm^n}x+a\paraa{\Fpn-\Fmn}=(\lp-\lm)y,\label{eq:Lhit2}
  \end{align}
  where $\Fpn = \paraa{1-\lp^n}/(1-\lp)$ and
  $\Fmn=\paraa{1-\lm^n}/(1-\lm)$. In the current case, $0<\lm<1$ and
  $\lp>1$, which implies that $\Fpn>\Fmn>0$. Thus, when $a\geq 0$ we
  must have $x\leq 0$ by equation \eqref{eq:Lhit1}. 
\end{proof}

\begin{lemma}\label{lemma:nostring2}
  Assume that $\det A=1$, $\Delta=0$ and $a\geq 0$. For every integer
  $n\geq 1$ there exist no $x,y>0$ such that $\Lh^n(x,0)=(0,y)$.
\end{lemma}

\begin{proof}
  When $\det A=1$ and $\tr A=2$, the $n$'th iterate of the affine map
  can easily be calculated as
  \begin{align*}
    \Lh^n\twovec{x}{y}=
    \twomatrix{1+n}{-n}{n}{1-n}\twovec{x}{y}+\frac{an}{2}\twovec{n+1}{n-1},
  \end{align*}
  and one sees directly that $\Lh^n(x,0)=(0,y)$ implies that $x\leq 0$ since
  $a\geq 0$.
\end{proof}

\begin{proposition}\label{prop:CLaOneDimRep}
  Assume that $\det A=1$ and $\Delta\leq 0$. If $\phi$ is an
  irreducible finite dimensional $\ast$-representation of $\CLa$ in
  one of the following situations
  \begin{enumerate}
  \item $a\geq 0$,
  \item $\Delta<0$, $a<0$ and $\ch\leq 0$,
  \item $\Delta<0$, $a<0$, $\ch>0$ and $\mu/\sqrt{\ch}\leq 1$,
  \end{enumerate}
  then $\phi$ is one-dimensional.
\end{proposition}

\begin{proof}
  In all three cases, Lemma \ref{lemma:LhNoPeriods} implies that there
  can be no non-trivial (i.e. of dimension greater than one) loop
  representations.

  In Case 1, Lemma \ref{lemma:nostring1} and Lemma
  \ref{lemma:nostring2} imply that there are no non-trivial string
  representations.  In Case 2 one can explicitly check that there
  is no component of $\Gamma$ intersecting both positive axes. Thus,
  no non-trivial string representations can exist.  In Case 3, there
  are constraint curves with a component that do intersect both
  positive axes. However, one can explicitly check that iterations of
  the point of intersection with the positive $r$-axis (where a string
  must start) increases the $r$-coordinate. Thus, one can never hit
  the positive $s$-axis (where a string must end) which implies that
  no non-trivial string representations can exist.
\end{proof}

\noindent In \cite{abhhs:noncommutative}, a special case of $\CLa$ was
considered where it holds that $-2<\tr A<2$. Then $\tr A$ can be
parametrized by setting $\tr A = 2\cos2\theta$. This makes it obvious
that the affine map $\Lh$ corresponds to a ``rotation'' by $2\theta$ on
the constraint curve, which will be an ellipse. The same kind of
parametrization can be done when $\tr A > 2$, in which case it is
convenient to set $\tr A = 2\cosh2\theta$ with $\theta>0$. Let us
gather the formulas one obtains in the following proposition.

\begin{proposition}\label{prop:CLaThetaParam}
  Assume that $\det A=1$, $\ch>0$ and $\tr A=2\cosh2\theta$ for
  some $\theta>0$. If we set
  \begin{align*}
    \xv_1(\beta) &= \sqrt{\ch}\parac{\frac{\mu}{\sqrt{\ch}}+\frac{\cosh\beta}{\cosh\theta},
    \frac{\mu}{\sqrt{\ch}}+\frac{\cosh(\beta-2\theta)}{\cosh\theta}},\\
    \xv_2(\beta) &= \sqrt{\ch}\parac{\frac{\mu}{\sqrt{\ch}}-\frac{\cosh\beta}{\cosh\theta},
    \frac{\mu}{\sqrt{\ch}}-\frac{\cosh(\beta-2\theta)}{\cosh\theta}},
  \end{align*}
  then the following holds
  \begin{enumerate}
  \item $\Gamma=\{\xv_1(\beta): \beta\in\reals\}\cup\{\xv_2(\beta):\beta\in\reals\}$,
  \item $\Lh\paraa{\xv_i(\beta)}=\xv_i(\beta+2\theta)$ for $i=1,2$,
  \item $\Lh^{n-1}(x,0)=(0,x)$ if and only if
    \begin{equation*}
      x=\frac{2\mu\sinh\theta\sinh(n-1)\theta}{\cosh n\theta}.
    \end{equation*}
  \end{enumerate}
\end{proposition}

\section{$\ast$-representations of $\A_L$}\label{sec:repsAL}

\noindent To study the representation theory of $\A_L$, we will use
the same techniques as for the representation theory of $\CLa$; we will
again see that the dynamical properties of an affine map is of crucial
importance. Since there exists a homeomorphism $\psi:\A_L\to\CLa$,
every representation of $\CLa$ induces a representation of
$\A_L$. However, in general there are representations that can not be
induced from $\CLa$.

In any (finite dimensional) $\ast$-representation $\phi$, $\phi(E)$
and $\phi(\Et)$ will be two commuting hermitian matrices. Therefore,
any such representation is unitarily equivalent to one where both
$\phi(E)$ and $\phi(\Et)$ are diagonal. Let us assume such a basis to
be chosen and write
\begin{align*}
  &E = \diag(e_1,e_2,\ldots,e_n)\\
  &\Et = \diag(\et_1,\et_2,\ldots,\et_n).
\end{align*}
For matrices in this basis, the defining relations of $\A_L$ reduce to
\begin{align*}
  &S_{ij}\paraa{\alpha e_i+\beta\et_i+u-e_j}=0\\
  &S_{ij}\paraa{\gamma e_i+\delta\et_i+v-\et_j}=0,
\end{align*}
since $S^\dagger = T$ and $E$, $\Et$ are diagonal. There are two ways
of fulfilling these equations: Either $S_{ij}=0$ or
\begin{align*}
  \twovec{e_j}{\et_j} = \Amatrix\twovec{e_i}{\et_i}+\twovec{u}{v} = L\twovec{e_i}{\et_i},
\end{align*}
and by defining $\vv_i=(e_i,\et_i)$ we write this as $\vv_j=L(\vv_i)$.

Let $G_S=(V,E)$ be the directed graph of $S$.  If $(i,j)\in E$
(i.e. $S_{ij}\neq 0$) then a necessary condition for a representation
to exist is that $\vv_j=L(\vv_i)$. On the other hand, given a graph
$G$ and vectors $\{\vv_k\}$ such that $\vv_j=L(\vv_i)$ if $(i,j)\in
E$, then any matrix whose digraph equals $G$ defines a representation
of $\A_L$. Hence, the set of representations can be parameterized by
graphs allowing such a construction.

\begin{definition}
  A graph $G=(\{1,2,\ldots,n\},E)$ is called \emph{$L$-admissible} if
  there exists $\vv_k\in\reals^2$ for $k=1,2,\ldots,n$, such that
  $\vv_j=L(\vv_i)$ if $(i,j)\in E$. An $L$-admissible graph is called
  \emph{nondegenerate} if there exists such a set
  $\{\vv_1,\ldots,\vv_n\}$ with at least two distinct vectors;
  otherwise the graph is called \emph{degenerate}.
\end{definition}

\noindent By this definition, the digraph of $S$ in any representation
is $L$-admissible, and every $L$-admissible graph generates at least
one representation. Clearly, given an $L$-admissible graph, there can
exist a multitude of inequivalent representations associated to it.
If $L$ has a fix-point $(e_f,\et_f)$, then any graph is $L$-admissible
and this representation corresponds to $E=e_f\mid$ and $\Et=\et_f\mid$
and $S$ an arbitrary matrix. However, not all graphs will be
nondegenerate $L$-admissible graphs.

Let us now show that in the case when the representation is locally
injective (cp. Remark \ref{rem:locallyInjective}), we can bring it to
a convenient form. Let $G=(V,E)$ be an $L$-admissible connected graph
(if it is not connected, the representation will trivially be
reducible, and we can separately consider each component) and let $S$
be a matrix with digraph equal to $G$, such that the representation is
locally injective. Furthermore, let $\vt_1,\ldots,\vt_k$ be an
enumeration of the pairwise \emph{distinct} vectors in the set
$\{\vv_1,\ldots,\vv_n\}$ such that $\vt_{i+1}=L(\vt_i)$, and define
$V_i\subseteq V$ as follows
\begin{align*}
  V_i = \{l\in V: \vv_l=\vt_i\}\qquad\text{for }i=1,\ldots,k.
\end{align*}
Since the representation is assumed to be locally injective, we can
only have edges from vertices in the set $V_i$ to vertices in the set
$V_{i+1}$ (identifying $k+1\equiv 1$). Hence, the vertices of the
graph can be permuted such that the matrix $S$ takes the following
block form:
\begin{align}
  S = 
  \begin{pmatrix}
    0               &  S_1  & 0                & \cdots & 0 \\
    0               &  0                & S_2  & \cdots & 0 \\
    \vdots          & \vdots            & \ddots           & \ddots & \vdots\\
    0               & 0                 & \cdots           & 0 &  S_{k-1} \\
    S_k & 0                 & \cdots           & 0 & 0
  \end{pmatrix}\label{eq:SRep}    
\end{align}
where each matrix $S_i$ is a $|V_i|\times |V_{i+1}|$ matrix.  Thus,
the representations of $\A_L$ are generated by the affine map $L$ in
the following way: Any point $\vt_1\in\reals^2$ gives rise to the
points $\vt_i=(e_i,\et_i)=L^{i-1}(\vt_1)$; by setting
\begin{align*}
  E = 
  \begin{pmatrix}
    e_1\mid_{n_1} & & \\
    & \ddots & \\
    & & e_k\mid_{n_k} 
  \end{pmatrix}
  \qquad
  \Et = 
  \begin{pmatrix}
    \et_1\mid_{n_1} & & \\
    & \ddots & \\
    & & \et_k\mid_{n_k} 
  \end{pmatrix}
\end{align*}
together with any matrix of the form \eqref{eq:SRep}, with $S_i$ a
$n_i\times n_{i+1}$ matrix, one obtains a representation of $\A_L$ of
dimension $n_1+\cdots+n_k$. Unless $\xv_1$ is a periodic point of
order $k$ we must set $S_k=0$. Moreover, distinct iterations of $L$
(i.e, at least one of the points differ) can not give rise to
equivalent representations since the eigenvalues of $E$ and $\Et$ will
be different.

\section{Representations and surface geometry}\label{sec:repsurf}

\noindent We will now study the relation between the geometry of the
inverse image $\Sigma=C^{-1}(0)$ and representations of the derived
algebra $\CLa$. More precisely, the geometry of $C^{-1}(0)$, for
different values of $\alpha_0,\alpha_1,c_0$ will be compared with the
representations of $\CLa$ with $\ch_0$ (the value of the central
element) being equal to $c_0$, and $a,\tr A, \det A$ related to
$\alpha_0,\alpha_1,\betat_1,\gammat_1,\deltat_1,\hbar$ as in Section
\ref{sec:SurfaceRelation}. Furthermore, the comparison will be made
for small positive values of $\hbar$. When $\det A\neq 0$, the affine
map $\Lh$ will be invertible, and therefore Theorem \ref{thm:CLaRep}
applies, i.e. all finite dimensional $\ast$-representations can be
classified in terms of loops and strings.

Let us rewrite the polynomial $C(x,y,z)$, as defined in
\eqref{eq:Cstd}, to a form which makes it easier to identify the
topology of the surface in the case when $\alpha_1\neq 0$
\begin{equation*}
  C(x,y,z) = \frac{\alpha_1}{4}\bracketd{
    \parac{x^2+y^2+\frac{\alpha_0}{\alpha_1}}^2+\frac{2}{\alpha_1}z^2
    -\parac{\frac{\alpha_0^2}{\alpha_1^2}+\frac{2c_0}{\alpha_1}}
  }.
\end{equation*}
If $\alpha_1<0$ the inverse image will be non-compact, but if
$\alpha_1>0$ the genus of the surface will be determined by the
quotient $\mu/\sqrt{c}$ where
\begin{equation*}
  \mu = -\frac{\alpha_0}{\alpha_1}\qquad\text{and}\qquad
  c = \frac{\alpha_0^2}{\alpha_1^2}+\frac{2c_0}{\alpha_1}.
\end{equation*}
If $-1<\mu/\sqrt{c}<1$ the inverse image is a compact surface of genus
0, and if $\mu/\sqrt{c}>1$ the surface has genus 1 (see
\cite{abhhs:noncommutative} for details and proofs).  When $\alpha_1=0$, the
polynomial becomes
\begin{align*}
  C(x,y,z) = \frac{\alpha_0}{2}\paraa{x^2+y^2}+\half z^2 - \half c_0,
\end{align*}
and the smooth inverse images consist of ellipsoids and (one or two
sheeted) hyperboloids. A complete table of the different geometries
can be found in Appendix B. We note that when the algebra $\CLa$
arises from a surfaces, then $\mu^2+2\ch_0/\alpha_1=\ch$.

By introducing
\begin{align*}
  t^2 = \frac{1+2\hbar^2\deltat_1-\frac{1}{2}\hbar^2\alpha_1}{4\hbar^2}
  \quad\text{ and }\quad
  \ch = \frac{\alpha_0^2}{\alpha_1^2}+\frac{2\ch_0}{\alpha_1}
\end{align*}
one can rewrite the defining equation of the constraint curve $\Gamma$ as
\begin{align}
  &p(r,s)=\parab{r+s+\frac{2\alpha_0}{\alpha_1}}^2+\frac{8t^2}{\alpha_1}(r-s)^2-4\ch\label{eq:pNonzero}\\
  &p(r,s)=\frac{\alpha_0}{2}(r+s)+t^2(r-s)^2-\ch_0\label{eq:pZero}
\end{align}
when $\alpha_1\neq 0$ and $\alpha_1=0$ respectively. Since we only
consider small values of $\hbar$, we can assume that $t^2>0$.

Note that we will use the parameters of the algebra and the parameters
of the surface interchangeably, and they are assumed to be related as
in Section \ref{sec:SurfaceRelation}.

\subsection{The degenerate cases}

\noindent Let us take a look at the cases when the inverse image is
not a surface (P.1 --P6, Z.1 -- Z.4), by studying some examples. For
instance, in case P.4, $\Sigma$ will be the empty set, and we easily
see that there are no non-negative $(r,s)$ on the constraint curve
$\Gamma$. Therefore, since the eigenvalues of $D$ and $\Dt$ are
non-negative, no representations can exist.  In case P.2 one gets
$\Sigma=\{(0,0,0)\}$, and the only non-negative point on $\Gamma$ is
$(0,0)$. Therefore, all representations must satisfy $D=\Dt=0$, which
implies that $W=0$.

By considering all degenerate cases, one can compile the following
table:

\renewcommand{\arraystretch}{1.3}
\begin{center}
  \begin{tabular}{|l|l|}
    \hline
    $\Sigma=C^{-1}(0)$                           & Irreducible $\ast$-representations $\phi$ \\ \hline\hline
    $\emptyset$                                  & None \\ \hline
    $\{(0,0,0)\}$                                & $\phi(W)=0$ \\ \hline
    $\{(x,y,0): x^2+y^2=|\alpha_0|/|\alpha_1|\}$ & $\phi(W)=\sqrt{|\alpha_0|/|\alpha_1|}$ \\ \hline
  \end{tabular}
\end{center}
\renewcommand{\arraystretch}{1}

\noindent In particular, we note that all irreducible representations are one-dimensional.

\subsection{Compact surfaces}

\noindent We will focus on the compact surfaces for which
$\alpha_1>0$ (P.7 -- P.10), as the only other compact surface (Z.5) can be treated
analogously.  When $\alpha_1>0$ and $\ch>0$, the constraint curve will
be an ellipse symmetric around the line $\pi/4$ and centered at
$(\mu,\mu)$.  The analysis of the corresponding finite dimensional
representations was done in \cite{abhhs:noncommutative} but we will
recall some basic facts.

Let us introduce $\theta\in(0,\pi/2)$ such that
$2\cos2\theta=\tr A$. The action of $\Lh$ can then be written as
\begin{align*}
  \Lh\twovec{x}{y}=\twomatrix{2\cos 2\theta}{-1}{1}{0}\twovec{x}{y}+\twovec{4\mu\sin^2\theta}{0}
\end{align*}
and one can understand it as a ``rotation'' by an angle $2\theta$ on
the ellipse.  One can easily show that $\Gamma\subset\Rtwoplus$ when
$\mu/\sqrt{\ch}>1/\cos\theta$; thus, by Lemma \ref{lemma:transrec}, no
representation in this region can contain a string, and therefore all
irreducible representations must consist of a single loop. When
$\mu/\sqrt{\ch}\leq 1$ no loop representations can exist, since a too
large part of the ellipse is contained in
$\reals^2\backslash\Rtwoplus$. In the small region
$1<\mu/\sqrt{\ch}\leq 1/\cos\theta$ ($\cos\theta\to 1$ as $\hbar\to
0$) both strings and loops can exist. We call surfaces in this region
\emph{critical tori}; these surfaces have a very narrow hole through
them.

However, representations do no exist for all values of $\theta$ and
the following conditions must be fulfilled for a $n$-dimensional
representation to exist:
\begin{align*}
  &\text{Loop:}\qquad e^{2in\theta}=1 \\
  &\text{String:}\qquad \sqrt{\ch}\cos n\theta+\mu\cos\theta=0 \\
  &\text{String (Z.5):}\qquad \ch_0 = \frac{\alpha_0^2\hbar^2(n^2-1)}{4(1+2\hbar^2\deltat_1)}.
\end{align*}
In the case when $\mu/\sqrt{\ch}>1/\cos\theta$ one can have two-sided
infinite dimensional representations by letting $\theta$ be an
irrational multiple of $\pi$; this is not possible for the sphere.
Let us summarize the representations for compact surfaces in the
following table:

\renewcommand{\arraystretch}{1.5}
\begin{center}
  \begin{tabular}{|l|l|}
    \hline
    $\Sigma=C^{-1}(0)$     & Irreducible $\ast$-representations $\phi$ \\ \hline\hline
    Sphere                 & String representations \\ \hline
    Critical torus         & String and loop representations. \\ \hline
    (Non-critical) torus   & Loops, two-sided infinite representations.\\ \hline
  \end{tabular}  
\end{center}
\renewcommand{\arraystretch}{1}
\noindent As an example, let us construct an 11-dimensional loop
representation when the surface is a torus. More precisely, we set
$\theta=\pi/11$, $\hbar=\tan(\theta)$, $\tr A=2\cos 2\theta$, $a=1/2$,
$\ch=1$ and $\deltat_1=1/2$, which corresponds to $\alpha_0\approx
1.99$, $\alpha_1\approx 3.15$ and $c=1$. In Figure \ref{fig:looprep}
one finds the corresponding constraint curve and the 11 points of
iteration of the affine map $\Lh$. Let $\xv_1$
(e.g. $\approx(1.56,1)$) be an initial point on the curve and let
$\xv_k=(d_k,\dt_k)=L^{k-1}(\xv_1)$ be its iterations. A
$\ast$-representation of $\CLa$ is then constructed by setting
\begin{equation*}
  \phi(W) = 
  \begin{pmatrix}
    0               &  \sqrt{d_1}  & 0                & \cdots & 0 \\
    0               &  0                & \sqrt{d_2}  & \cdots & 0 \\
    \vdots          & \vdots            & \ddots           & \ddots & \vdots\\
    0               & 0                 & \cdots           & 0 &  \sqrt{d_{10}} \\
    \sqrt{d_{11}}               & 0                 & \cdots           & 0 & 0
  \end{pmatrix}
  \approx
  \begin{pmatrix}
    0               &  1.25  & 0                & \cdots & 0 \\
    0               &  0                & 1.46  & \cdots & 0 \\
    \vdots          & \vdots            & \ddots           & \ddots & \vdots\\
    0               & 0                 & \cdots           & 0 &  0.79 \\
    1.00               & 0                 & \cdots           & 0 & 0
  \end{pmatrix}.
\end{equation*}

\begin{figure}[t]
  \centering
  \includegraphics[height=4cm]{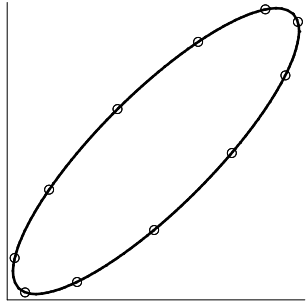}  
  \caption{The constraint curve and the points of iterations for an
    11-dimensional loop representation when $\alpha_0\approx 1.99$,
    $\alpha_1\approx 3.15$ and $c=1$.}
  \label{fig:looprep}
\end{figure}

\subsection{Non-compact surfaces}

\noindent The remaining surfaces will have one or two non-compact
components (except for the surfaces in Section \ref{sec:cptNoncpt},
which has both a compact and a non-compact component) and we will show
that infinite representations always exist, whereas all finite
dimensional representations are one-dimensional. By looking at the
tables in Appendix B, one sees that non-compact surfaces appear only
when $\alpha_1\leq 0$ which, for small $\hbar$, is equivalent to
$\Delta\leq 0$. We can now prove the following result about the
relation between geometry and representations.

\begin{proposition}\label{prop:NonCompact}
  Let $\CLa$ be an algebra  corresponding to a surface $\Sigma$
  where each component is non-compact, and assume that at least one of
  $\alpha_0,\alpha_1,c_0$ is  different from zero.  Then the following
  holds:
  \begin{enumerate}
  \item All finite dimensional irreducible representations have dimension one.
  \item If $\Sigma$ has two components then there exists two
    inequivalent one-sided infinite dimensional irreducible
    representations, but no two-sided representations.
  \item If $\Sigma$ is connected and non-singular, then there exists a
    two-sided infinite dimensional irreducible representation; if
    $a\leq 0$, or $a>0$ and $c>\mu^2(1+|\Delta|/4)$, then no one-sided
    representations exist. If $a>0$ and $c\leq\mu^2(1+|\Delta|/4)$
    then one-sided representations exist.
  \end{enumerate}
\end{proposition}

\begin{proof}
  Statement 1 follows immediately from Proposition
  \ref{prop:CLaOneDimRep}.  Statement 2 can be proved in the following
  way: By examining all cases in Appendix B where $C^{-1}(0)$ has two
  non-compact components (Z.8, N.1, N.2, N.5, N.8), one sees that the
  components of $\Gamma$ which intersect $\Rtwopluszero$ has the
  following form
  \begin{center}
    \includegraphics[width=2cm]{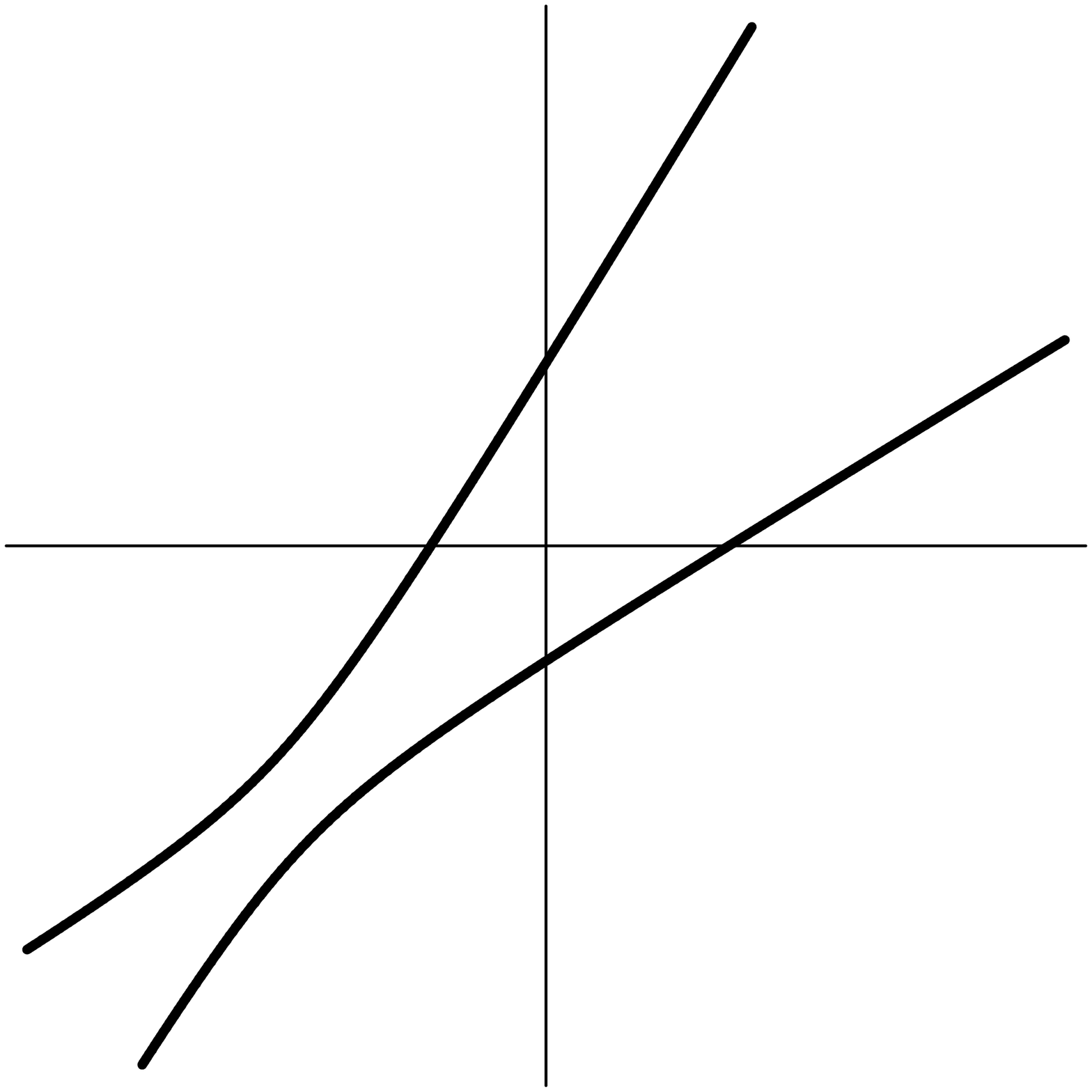}\hspace{20mm}
    \includegraphics[width=2cm]{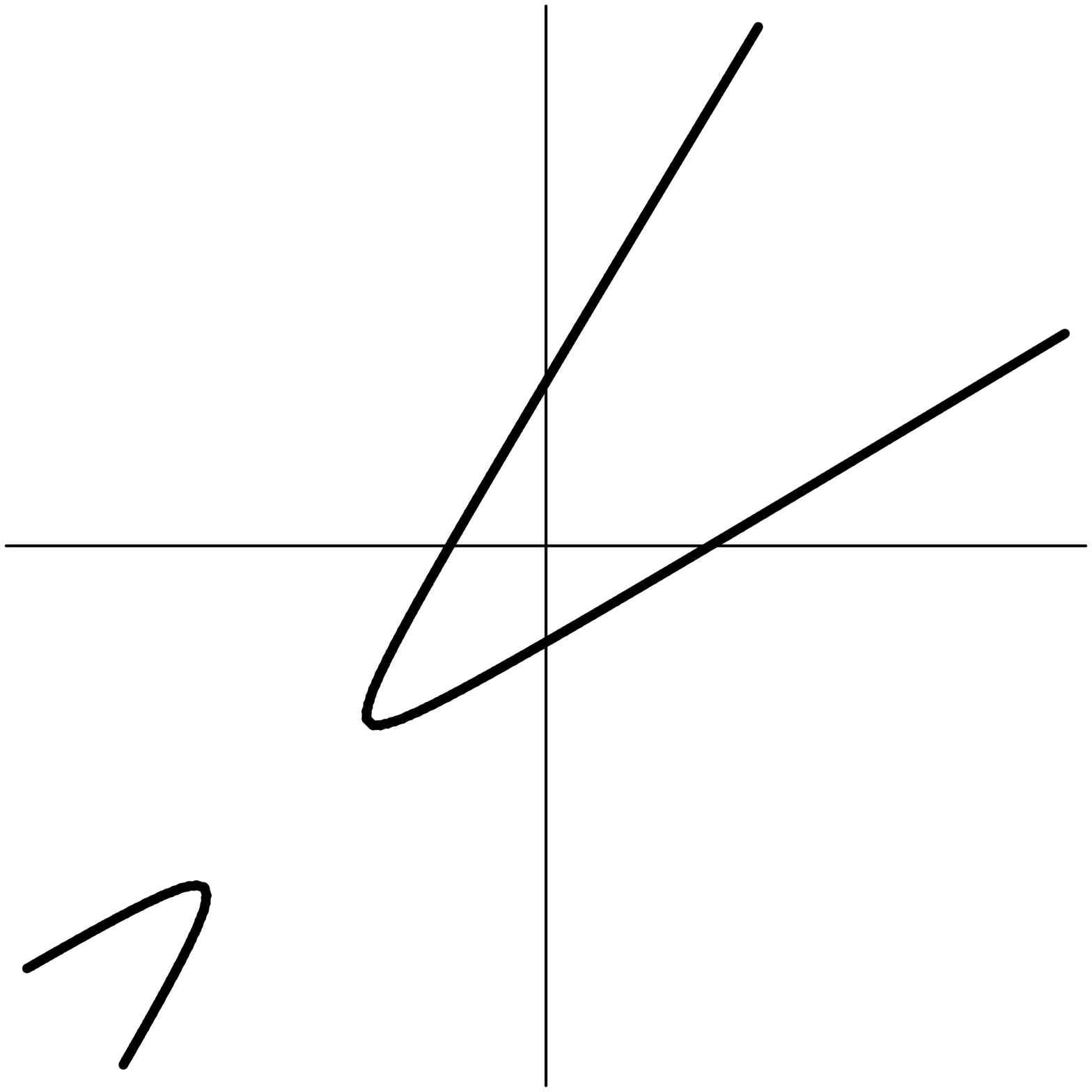}
  \end{center}
  In the first case, two inequivalent one-sided representations can be
  constructed but no two-sided representations can exist because
  backward or forward iterations of any point will eventually reach
  outside $\Rtwoplus$. In the second case, it holds that the lower
  left tip of the curve intersecting $\Rtwopluszero$ has strictly
  negative coordinates and the curve intersects the positive axes
  exactly once. This immediately allows for a construction of two
  inequivalent one-sided infinite dimensional representations. Now,
  can we have two-sided representations? Any constraint curve will
  cross the positive $r$-axis in the following points:
  \begin{equation}\label{eq:axesCrossing}
    r_{\pm} = \half\bracketb{a\pm\sqrt{a^2+4\ch_1}}
  \end{equation}
  If there is only one strictly positive intersection-point, it must
  hold that $r_+>0$ but $r_-\leq 0$. Actually $r_-<0$ since the lower
  left tip of the curve is not in $\Rtwopluszero$. A necessary
  condition for a two-sided representations to exist is that there
  exists a point on the curve such that all backward and forward
  iterations by $\Lh$ is contained in $\Rtwoplus$. This means that
  (since $\Lh$ preserves the direction of the curve) when we apply
  $\Lh$ to the point of intersection with the $s$-axis, we must obtain
  a point in $\Rtwopluszero$ (otherwise no point is able to ``jump''
  the negative part of $\Gamma$ by the action of $\Lh$). But this does
  not happen since
  \begin{align*}
    \Lh\twovec{0}{r_+}=\twovec{-r_++a}{0}=\twovec{r_-}{0}\notin\Rtwopluszero.
  \end{align*}
  Let us now prove the statement 3. When $a\leq 0$ the are only two
  cases which give a connected non-singular non-compact surface,
  namely N.7 and N.10. In both cases, one can check that $\Gamma$ does
  not intersect the axes, and that at least one component is contained
  in $\Rtwoplus$. Hence, no one-sided representations exist but
  two-sided representations exist.

  When $a>0$ (Z.6, N.4) and $c>\mu^2(1+|\Delta|/4)$ then the component
  of $\Gamma$ that intersects $\Rtwopluszero$ is contained in
  $\Rtwoplus$, which implies that no one-sided representations exist,
  but two-sided representations exist.  When $a>0$ and
  $c\leq\mu^2(1+|\Delta|/4)$ one component of the constraint curve
  will have the following form
  \begin{center}
    \includegraphics[width=2cm]{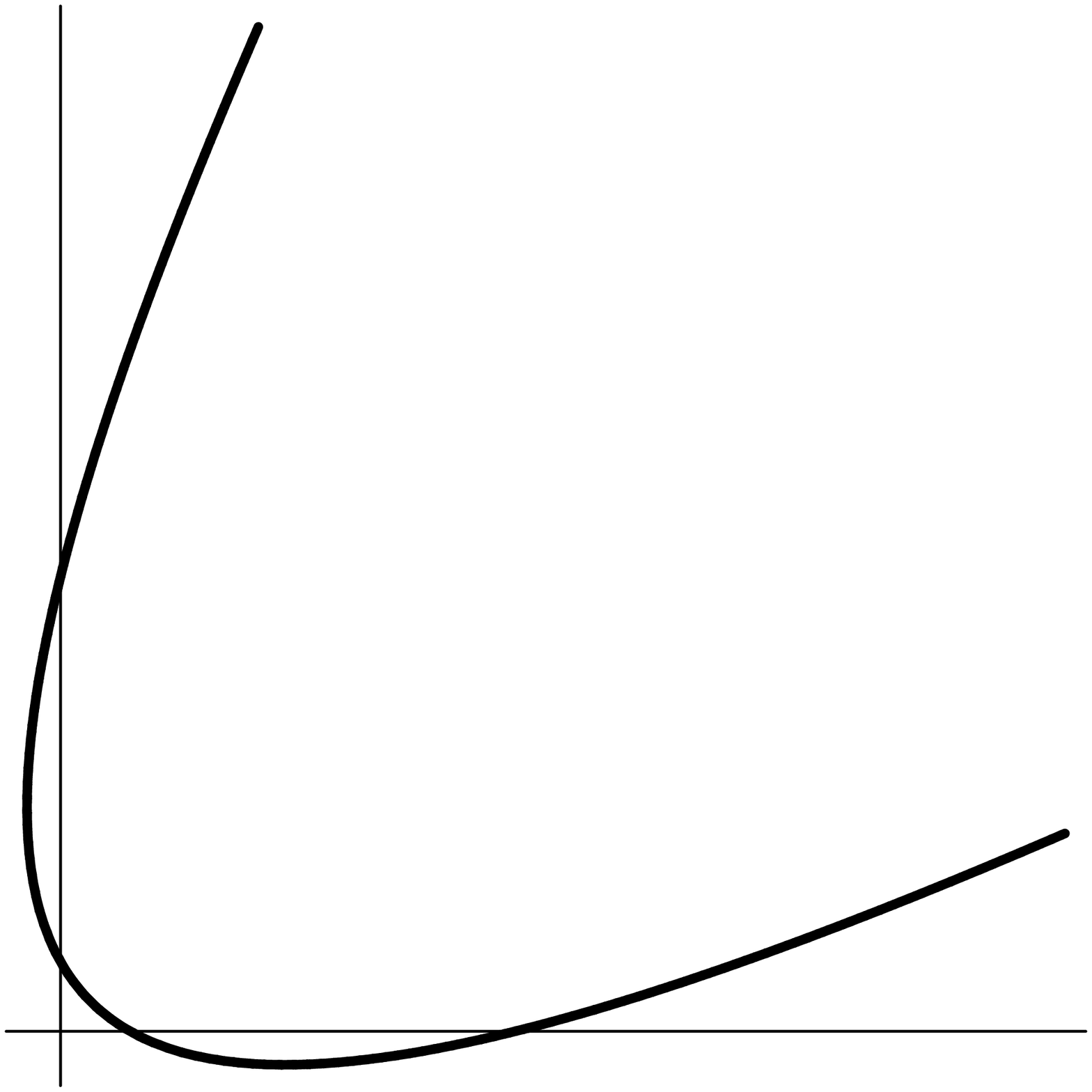}
  \end{center}
  In particular, it intersects the positive axes at least once. Thus,
  one-sided representations can be easily defined, but what about
  two-sided representations? We will now show that the backward and
  forward iteration of the point at the lower left tip both lie in
  $\Rtwoplus$. We consider only the case Z.6 as the other case (N.4)
  can be treated analogously. The lower tip of the ellipse has
  coordinates $(r_0,r_0)$ for some $r_0>0$. We calculate
  \begin{align*}
    \Lh\twovec{r_0}{r_0}=\twovec{r_0+a}{r_0}\in\Rtwoplus\quad\text{ and }\quad
    \Lh^{-1}\twovec{r_0}{r_0} = \twovec{r_0}{r_0+a}\in\Rtwoplus,
  \end{align*}
  since $a>0$. Hence, we can define a two-sided representation by
  starting at $(r_0,r_0)$ and considering all backward and forward
  iterations by $\Lh$.
\end{proof}
\noindent The connected non-compact surfaces in Proposition
\ref{prop:NonCompact} which allow for one-sided infinite dimensional
representations correspond to surfaces with a narrow ``throat'' as in
Figure \ref{fig:narrowThroat}.
\begin{figure}[t]
  \centering
  \includegraphics[height=4cm]{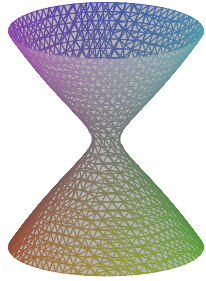}  
  \caption{A connected non-compact surfaces which allows for one-sided representations.}
  \label{fig:narrowThroat}
\end{figure}
This kind of ``tunneling'' is analogous to what happens for the
compact surfaces. For a compact torus with a narrow hole, string
representations can continue to exist; but as the hole grows wider
they cease to exist. For non-compact surfaces, one-sided
representations continue to exist even though the two components have
been joined together. As an example, let us consider a surface with
$\alpha_0=-1$, $\alpha_1=-1$ and $c=1.02$. Choosing $\hbar=0.3$,
$\deltat_1=\betat_1=-1/4$ and $\gammat_1=-1/2$ gives us the constraint
curve in Figure \ref{fig:onesided_twosided_rep}. On the left curve,
the first iterations of a one-sided infinite dimensional
representations are plotted; on the right curve one finds iterations
corresponding to a two-sided representation. The representations that
are defined by these two figures the following form:
\begin{align*}
  \phi_{\text{\small one-sided}}(W) \approx 
  \begin{pmatrix}
    0 & 0.41 &\\
    & 0 & 0.74 \\ 
    & & 0 & 1.11 \\ 
    & & & 0 & 1.53 \\ 
    & & & & \ddots & \ddots
  \end{pmatrix}
\end{align*}

\begin{align*}
  \phi_{\text{\small two-sided}}(W) \approx 
  \begin{pmatrix}
    \ddots & \ddots &\\
    & 0 & 0.45 \\ 
    & & 0 & 0.01 \\ 
    & & & 0 & 0.01 \\ 
    & & & & 0 & 0.45 \\ 
    & & & & & 0 & 0.79 \\ 
    & & & & & & \ddots & \ddots
  \end{pmatrix}.
\end{align*}

\begin{figure}[h]
  \centering
  \includegraphics[width=4cm]{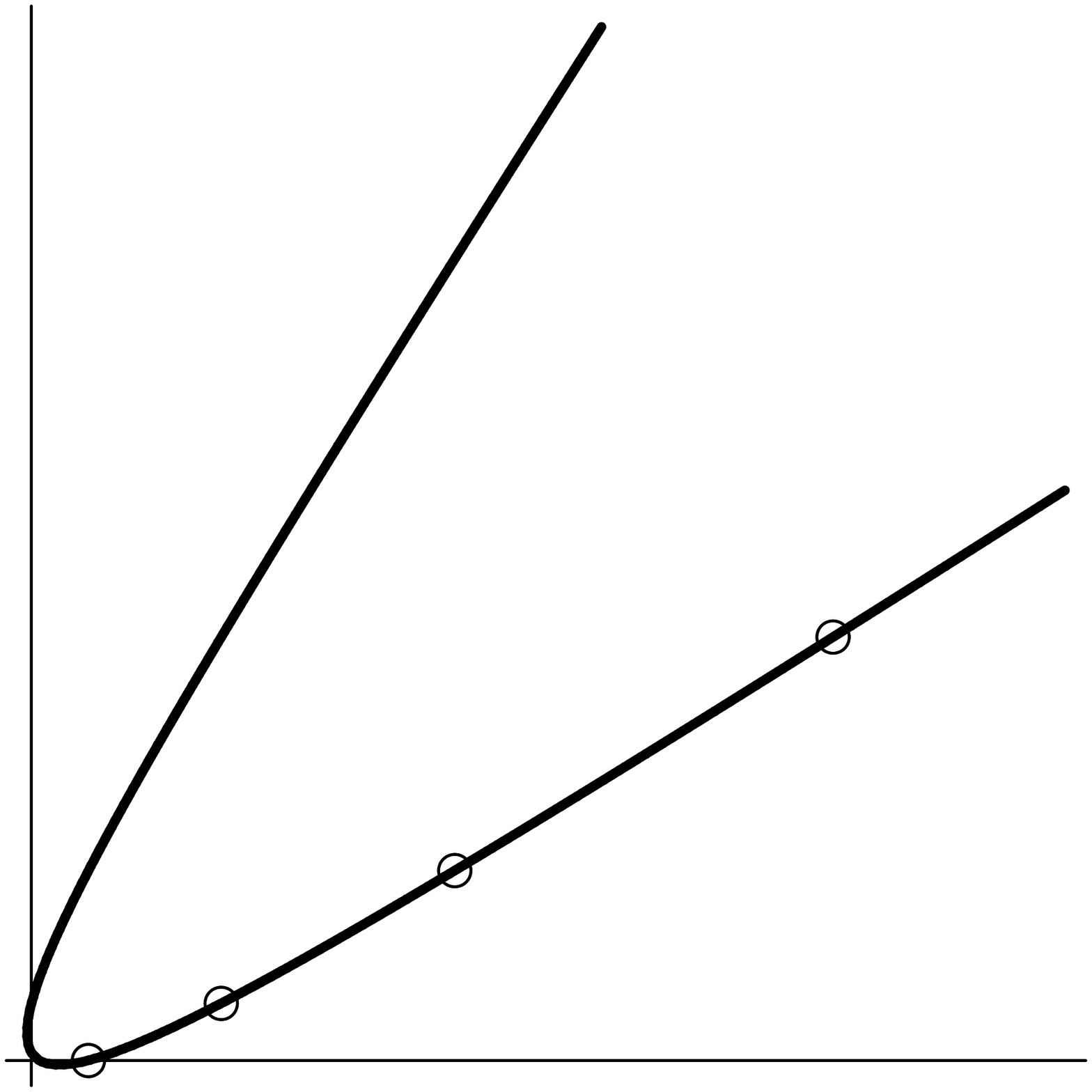}
  \hspace{20mm}
  \includegraphics[width=4cm]{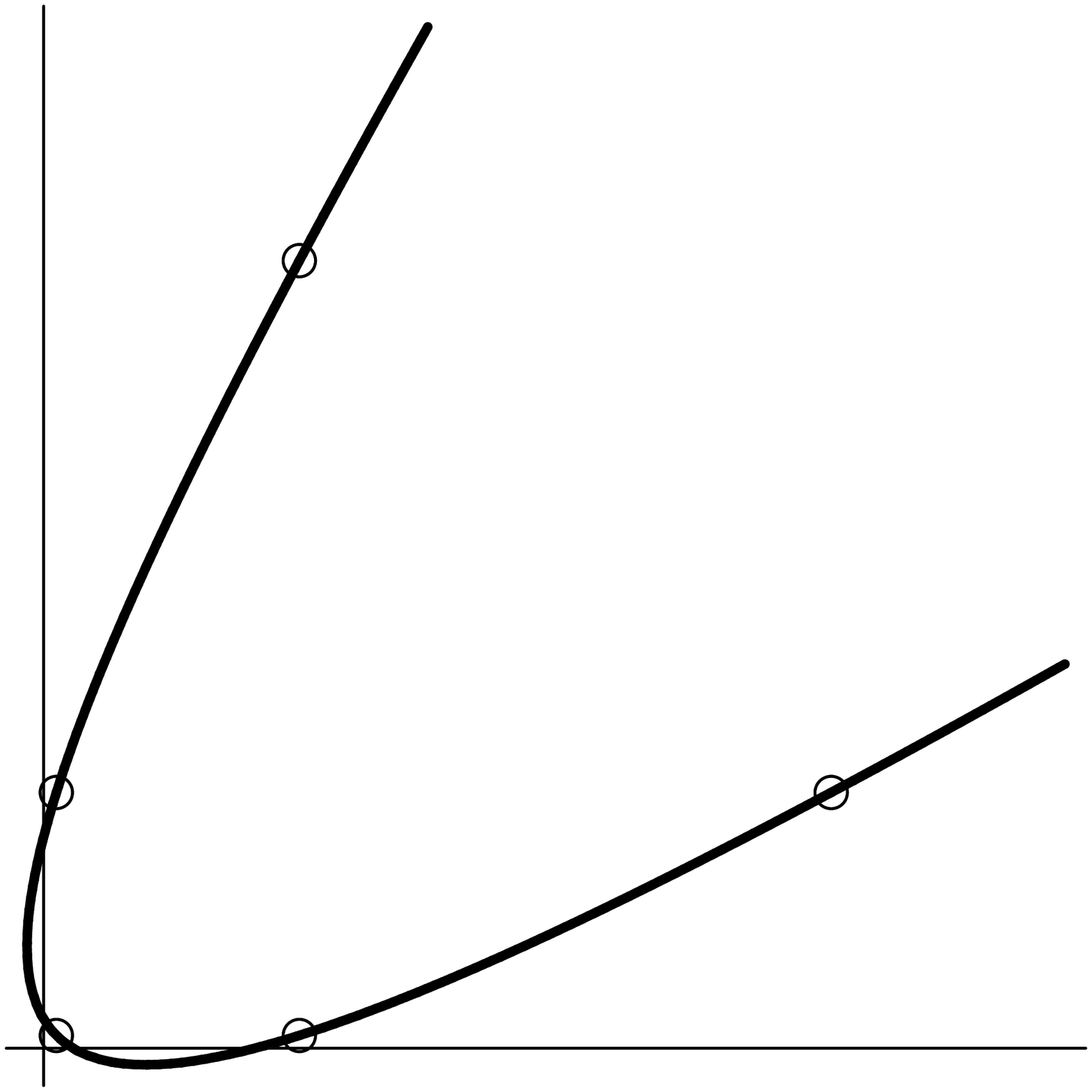}
  \caption{Constraint curves and iterations corresponding to a
    non-compact one-sheeted surfaces with $\alpha_0=-1$, $\alpha_1=-1$
    and $c=1.02$.}
  \label{fig:onesided_twosided_rep}
\end{figure}

\noindent Let us make a remark about the surface that has been excluded from
Proposition \ref{prop:NonCompact}, namely Z.2 with $c_0=0$. The
polynomial reduces to $C(x,y,z)=z^2$ and the inverse image describes
the $x,y$-plane. The constraint curve is the line $r-s=0$ and every
point $(r,r)$ is a fix-point of $\Lh$. Hence, all irreducible
representations are one-dimensional and the inequivalent
representations are parametrized by the non-negative real numbers.

\subsection{A surface with both compact and non-compact components}\label{sec:cptNoncpt}

In the cases N.11 and N.12, the surface consists of two components: a
non-compact surface and a compact surface of genus 0 (which collapses
to a point when $\mu/\sqrt{c}=1$). The constraint curve will have the
form as in Figure \ref{fig:cptNoncptSurface}.
\begin{figure}[h]
  \centering
  \includegraphics[height=4cm]{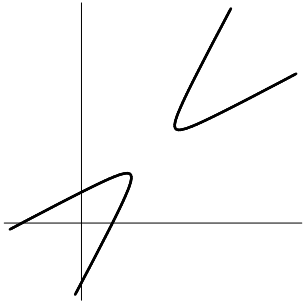}  
  \caption{The constraint curve corresponding to a surface with both compact and non-compact components.}
  \label{fig:cptNoncptSurface}
\end{figure}
Thus, there is one component which allows for the construction of a
finite dimensional string representation, and one component that
induces a two-sided infinite dimensional representation. When
$\mu/\sqrt{c}=1$, the lower component will intersect $\Rtwopluszero$
in the point $(0,0)$, which allows for a one-dimensional trivial
representation.

As for the compact surfaces, string representations do not exist for
all values of $\hbar$ (if we fix $\ch$). In the notation of
Proposition \ref{prop:CLaThetaParam}, the condition for the existence
of a $n$-dimensional string representation is
\begin{equation*}
  \mu\cosh\theta-\sqrt{\ch}\cosh n\theta=0.
\end{equation*}

\subsection{Singular non-compact surfaces}

\noindent The surfaces Z.7, N.3, N.6 have a singularity at one point
(which arises as two sheets come together) and the surface N.9 is a
limit case of N.12, where the sphere touches the non-compact
surface. The corresponding constraint curves will have one of the
forms in Figure \ref{fig:singSurfaces},
\begin{figure}[h]
  \centering
  \includegraphics[width=3cm]{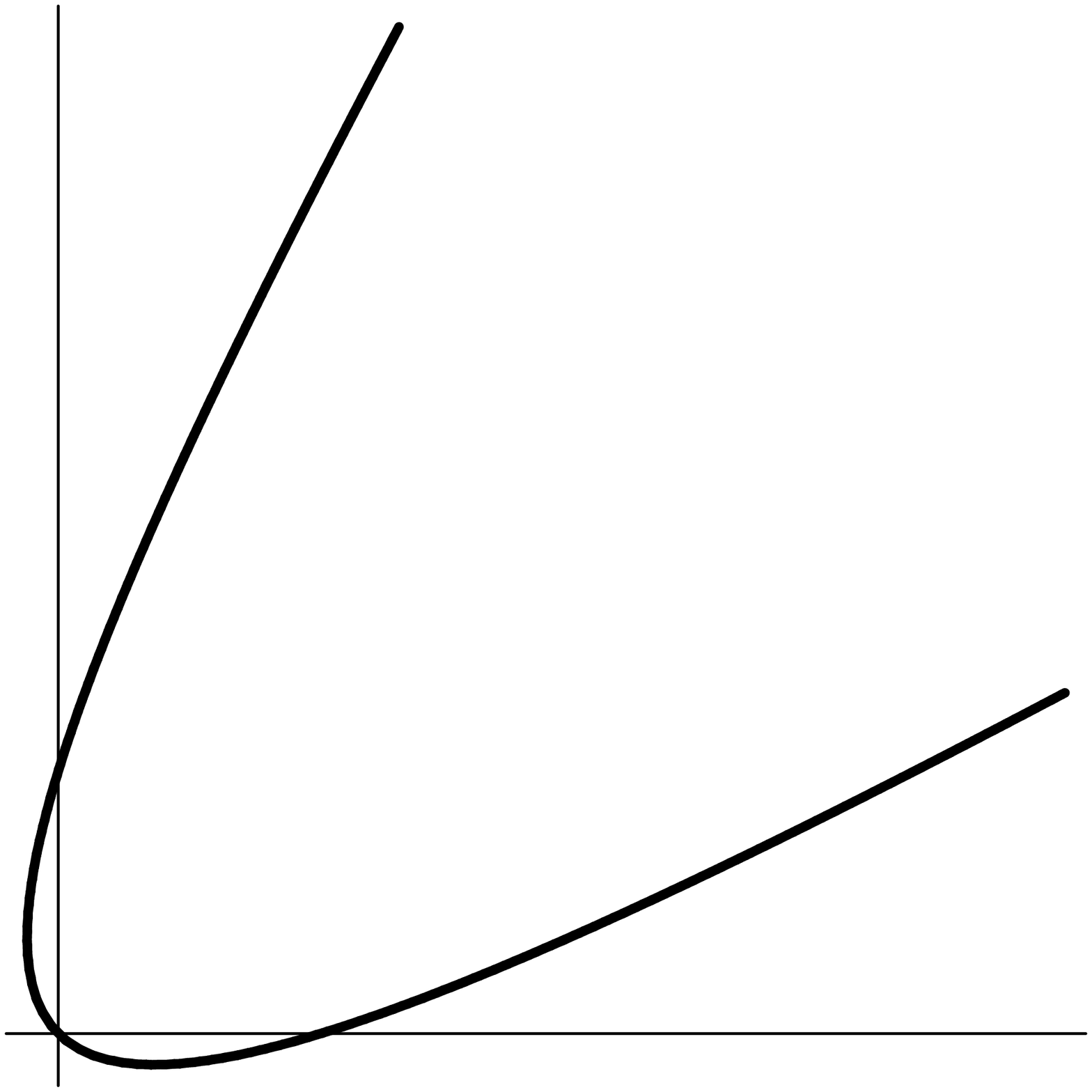}\hspace{20mm}
  \includegraphics[width=3cm]{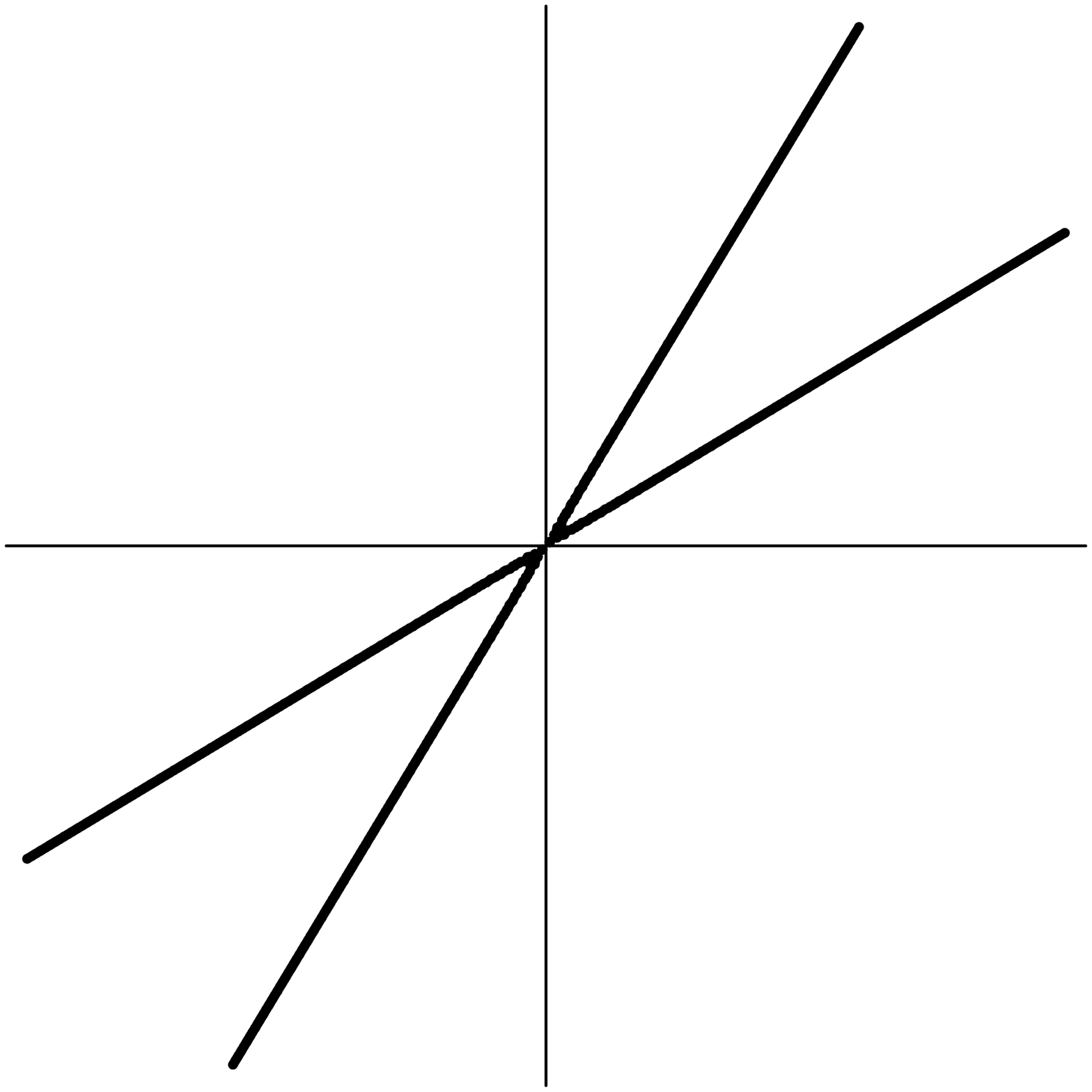}  
  \caption{Constraint curves of singular non-compact surfaces.}
  \label{fig:singSurfaces}
\end{figure}
where the left picture corresponds to Z.7 and N.3, and the right
picture corresponds to N.6 and N.9. The left constraint curve clearly
allows for two one-sided representations, but no two-sided
representation can exist (cp. proof of Proposition
\ref{prop:NonCompact}). The right constraint curve allows for two
different two-sided representations; it is easy to check that since
$\Lh$ is invertible and $(0,0)$ is a fix-point, all iterations of a
point on the curve in $\Rtwoplus$ stay in $\Rtwoplus$. That is,
iterations approach the origin but never reach it. Furthermore, we
note that no finite dimensional representations of dimension greater
than one can exist.

\subsection{Correspondence between geometry and representation theory}

\noindent In \cite{abhhs:noncommutative}, the representation theory of
$\CLa$ was compared with the geometry of the inverse image for a class
of compacts surfaces of genus 0 and 1. We have now extended this
analysis to inverse images of general rotationally symmetric fourth
order polynomials. Apart from recovering earlier results, we have
shown that the representation theory respects the geometry of the
surface to a high extent. Namely, in all cases where
$\Sigma=C^{-1}(0)$ is empty, no representations exist. When $\Sigma$
is not a surface, then no irreducible representations of dimension
greater than one exist. In the case when $\Sigma$ is non-singular
non-compact, the correspondence is as follows: if $\Sigma$ has two
sheets then there exists two inequivalent \emph{one-sided} infinite
dimensional representations, and if $\Sigma$ has one sheet there is a
\emph{two-sided} infinite dimensional representation. In all
non-compact cases no finite-dimensional representations of dimension
greater than one exist.

\section*{Acknowledgement}

\noindent This work was partially supported by the Swedish Research Council, the
Crafoord Foundation, the Swedish Royal Academy of Sciences and the
Swedish Foundation of International Cooperation in Research and Higher
Education (STINT). J. A. would also like to thank the Sonderforschungsbereich
SFB647 as well as the Institut des Hautes
{\'E}tudes Scientifiques for financial support and hospitality.

\newpage

\section*{Appendix A -- Solutions to system of equations}

\noindent The general solution to the four equations
\begin{align*}
  &\Amatrix
  \twomatrix{m}{n}{\mt}{\nt}-
  \twomatrix{m}{n}{\mt}{\nt}
  \twomatrix{\tr A}{-\det A}{1}{0}=0
\end{align*}
is given by
\begin{equation*}
  n = \beta\mt- \delta m\qquad
  \nt = \gamma m-\alpha\mt.
\end{equation*}
If $\Delta=1+\det A-\tr A\neq 0$ then the system
\begin{align}\label{eq:kkt}
  \bracketc{\Amatrix-\mid_2}
  \twovec{k}{\kt} = 
  \twomatrix{m}{n}{\mt}{\nt}\twovec{a}{0}-\twovec{u}{v}
\end{align}
has a unique solution for $k$ and $\kt$. Whenever $a\neq 0$, we can
always solve \eqref{eq:kkt} by setting
\begin{equation*}
  m=\frac{1}{a}\bracketb{(\alpha-1)k+\beta\kt+u}\qquad 
  \mt=\frac{1}{a}\bracketb{\gamma k+(\delta-1)\kt+v}.
\end{equation*}
If $\Delta=a=0$ there are two cases. When $A=\mid_2$ it is necessary
that $u=v=0$, in which case \eqref{eq:kkt} is identically
satisfied and the affine map $L$ will be the identity map. If
$A\neq\mid_2$ we get the following conditions
\begin{align*}
  &\alpha\neq 1:\qquad \text{if }(\alpha-1)v=\gamma u\text{ then }k=\frac{1}{1-\alpha}\paraa{u+\beta\kt}\\
  &\gamma\neq 0:\qquad \text{if }(\alpha-1)v=\gamma u\text{ then }k=-\frac{1}{\gamma}\paraa{v+(\delta-1)\kt}\\
  &\delta\neq 1:\qquad \text{if }(\delta-1)u=\beta v\text{ then }\kt=\frac{1}{1-\delta}\paraa{v+\gamma k}\\
  &\beta\neq 0:\qquad \text{if }(\delta-1)u=\beta v\text{ then }\kt=-\frac{1}{\beta}\paraa{u+(\alpha-1)k}.
\end{align*}

\newpage
\section*{Appendix B -- Inverse images of $C(x,y,z)$}

\renewcommand{\arraystretch}{1.3}
\begin{flushleft}
\begin{tabular}{|l|c|c|c||l|l|}
  \hline
  $\alpha_1>0$ & $\alpha_0$ & $c$     & $\mu/\sqrt{c}$     &$C^{-1}(0)$ & $\Gamma\cap\Rtwopluszero$\\ \hline 
  P.1 & --         & $<0$    & --       & $\emptyset$ & $\emptyset$\\ \hline
  P.2 & $0$        & $0$     & --       & $\{(0,0,0)\}$ & $\{(0,0)\}$\\ \hline
  P.3 & $<0$       & $0$     & --       & $\{(x,y,0): x^2+y^2=|\alpha_0|/\alpha_1\}$ & $\left\{\frac{|\alpha_0|}{\alpha_1}(1,1)\right\}$\\ \hline\hline
  P.4 & $>0$       & $0$     & --       & $\emptyset$ & $\emptyset$\\ \hline
  P.5 & $>0$       & $>0$    & $<-1$    & $\emptyset$ & $\emptyset$ \\ \hline
  P.6 & $>0$       & $>0$    & $-1$     & $\{(0,0,0)\}$ & $\{(0,0)\}$\\ \hline
  P.7 & $\geq 0$       & $>0$    & $>-1$    & Sphere & $\{\text{Ellipse}\}\cap\reals^+_0$ \\ \hline\hline
  P.8 & $<0$       & $>0$    & $<1$     & Sphere & $\{\text{Ellipse}\}\cap\reals^+_0$ \\ \hline
  P.9 & $<0$       & $>0$    & $1$      & Surface with singularity & $\{\text{Ellipse}\}\cap\reals^+_0$ \\ \hline
  P.10 & $<0$       & $>0$    & $>1$     & Torus & $\{\text{Ellipse}\}\cap\reals^+_0$ \\ \hline
\end{tabular}

\vspace{4mm}

\renewcommand{\arraystretch}{1.3}
\begin{tabular}{|l|c|c||l|l|} 
  \hline
  $\alpha_1=0$ & $\alpha_0$ & $c_0$     & $C^{-1}(0)$ & $\Gamma\cap\Rtwopluszero$\\ \hline
  Z.1 & $0$        & $<0$      & $\emptyset$ & $\emptyset$\\ \hline
  Z.2 & $0$        & $\geq 0$  & $\{(x,y,\sqrt{c_0})\}\cup\{x,y,-\sqrt{c_0}\}$ & Non-compact.\\ \hline\hline
  Z.3 & $>0$       & $<0$      & $\emptyset$ & $\emptyset$\\ \hline
  Z.4 & $>0$       & $0$       & $\{(0,0,0)\}$ & $\{(0,0)\}$\\ \hline
  Z.5 & $>0$       & $>0$      & Sphere & Compact. \\ \hline\hline
  Z.6 & $<0$       & $<0$      & One sheeted hyperboloid & Non-compact.\\ \hline
  Z.7 & $<0$       & $0$       & Singular hyperboloid & Non-compact.\\ \hline
  Z.8 & $<0$       & $>0$      & Two sheeted hyperboloid & Non-compact.\\ \hline
\end{tabular}

\vspace{4mm}

\begin{tabular}{|l|c|c|c||l|} 
  \hline
  $\alpha_1<0$ & $\alpha_0$   & $c$   & $\mu/\sqrt{c}$  & $C^{-1}(0)$ \\ \hline
  N.1 & $<0$         & $\leq0$  & --       & Two sheeted cone. \\ \hline
  N.2 & $<0$         & $>0$  & $<-1$     & Two sheeted cone. \\ \hline
  N.3 & $<0$         & $>0$  & $-1$      & One sheeted singular cone.  \\ \hline
  N.4 & $<0$         & $>0$  & $>-1$     & One sheeted cone.  \\ \hline\hline
  N.5 & $0$         & $<0$  & --       & Two sheeted cone.  \\ \hline
  N.6 & $0$         & $0$   & --       & One sheeted singular cone.  \\ \hline
  N.7 & $0$         & $>0$  & --       & One sheeted cone.  \\ \hline\hline
  N.8 & $>0$        & $<0$  & --       & Two sheeted cone. \\ \hline
  N.9 & $>0$        & $0$   & --       & One sheeted cone $\cup$ sphere (singular). \\ \hline
  N.10 & $>0$        & $>0$  & $<1$     & One sheeted cone.  \\ \hline
  N.11 & $>0$        & $>0$  & $1$      & One sheeted cone $\cup$ $\{(0,0,0)\}$. \\ \hline
  N.12 & $>0$        & $>0$  & $>1$     & One sheeted cone $\cup$ sphere.  \\ \hline
\end{tabular}
\end{flushleft}

\vspace{3mm}

\bibliographystyle{alpha}
\bibliography{crossedalg}

\end{document}